\documentclass[11pt,a4paper]{amsart}

\pdfoutput=1

\usepackage{geometry}                
\usepackage{graphicx}
\usepackage{amssymb}
\usepackage{epstopdf}
\usepackage{fullpage}
\usepackage{color}
\usepackage{float}
\usepackage{caption}
\usepackage[table]{xcolor}
\usepackage{tablefootnote}

\usepackage{url}

\RequirePackage[OT1]{fontenc}
\RequirePackage{amsthm,amsmath}
\usepackage{amsmath, amssymb, amsfonts} 
\usepackage{algorithm}
\usepackage{algorithmic}
\usepackage{booktabs}  
\usepackage{graphicx, subfigure} 
\usepackage{multirow} 

\usepackage{mathabx}

\usepackage{longtable}

\newtheorem{theorem}{Theorem}[section]

\newcommand{ \xb }{ \mathbf{x} } 

\newcommand{\bfX} {\mathbf{X}}

\newcommand{\bfx} {\mathbf{x}}

\renewcommand{\Pr}{\mathsf{P}}

\newcommand{\F}{\mathsf{F}}

\newcommand{\reals}{\mathbb{R}}

\newcommand{\I}{\mathsf{I}}

\newcommand{\nbd}{\mathsf{nbd}}

\newcommand{\Dir}{\mathsf{Dir}}

\renewcommand{\Pr}{\mathsf{P}}

\newcommand{\Ehat}{\overline{E}}

\DeclareMathOperator*{\di}{\mathrm{d}\!}

\begin{document}

\title{Loglinear model selection and human mobility}
\author{Adrian Dobra}
\address{Department of Statistics, Department of Biobehavioral Nursing and Health Informatics, Center for Statistics and the Social Sciences and Center for Studies in Demography and Ecology, University of Washington, Box 354322, Seattle, WA 98195}
\email{adobra@uw.edu}

\author{Abdolreza Mohammadi}
\address{Department of Methodology and Statistics, Tilburg University, Netherlands}
\email{a.mohammadi@uvt.nl}

\date{\today}                                           

\begin{abstract}
Methods for selecting loglinear models were among Steve Fienberg's research interests since the start of his long and fruitful career. After we dwell upon the string of papers focusing on loglinear models that can be partly attributed to Steve's contributions and influential ideas, we develop a new algorithm for selecting graphical loglinear models that is suitable for analyzing hyper-sparse contingency tables. We show how multi-way contingency tables can be used to represent patterns of human mobility. We analyze a dataset of geolocated tweets from South Africa that comprises $46$ million latitude/longitude locations of $476,601$ Twitter users that is summarized as a contingency table with $214$ variables.
\end{abstract}

\maketitle

\tableofcontents

\section{Introduction}
\label{sec:intro}

Steve Fienberg was one of the founders of modern multivariate categorical data analysis. In two of the books he wrote early in his career \cite{bishop-et-1975,fienberg-1980} he laid out key notation, definitions, modeling techniques, and also open research directions for building approaches for analyzing contingency tables. More than forty years ago, he argued that interactions of various orders among categorical variables are of great interest \--- a fact that is now recognized in the literature from several fields (e.g., biology, social sciences, public health, transportation research). Hierarchical loglinear models that represent log expected cell counts as sums of main effects of variables cross-classified in a table, and interactions of two, three or more of these variables are well suited to capture complex multivariate patterns of dependencies. The selection of the interaction structure in hierarchical loglinear models was a problem Steve discussed in considerable length in \cite[Chapter 9]{bishop-et-1975}, \cite[Chapter 4]{fienberg-1980}, and also in several papers he subsequently published later on in his career.

\cite{fienberg-1970} laid out one of the first strategies for hierarchical loglinear model determination which is based on partitioning the Pearson or the likelihood-ratio goodness-of-fit statistics into several additive parts. Steve's approach starts with a hierarchy of models, and a significance level. Interactions are sequentially added or deleted based on a series of tests that correspond with the partitioned components of the most complex models. The model search stops when the difference between consecutive models is significant. Steve properly recognized that a good model building strategy must walk the fine line between goodness-of-fit and parsimony, that is, including more interactions to obtain a better fit of the data, and leaving fewer interactions in the model to create simpler representations of the association structure. However, this early method for loglinear model selection can compare only models that are nested (i.e., a simpler model is obtained from a more complex one by deleting interactions), and can be successfully used for datasets that involve no more than $5$ variables.

Due in part to Steve's early contributions and ideas, several approaches to selection of loglinear models have started to emerge \cite{edwardshavranek1985,agresti-1990,whittaker1990}, but these methods turned out to be quite ineffective even for contingency tables with $7$ variables. One bottleneck is due to the exponential increase in the number of possible hierarchical loglinear models: while there are 7580 models with $5$ variables, there are about $5.6\times 10^{22}$ models with $8$ variables \cite{dellaportasforster1999}. Moreover, contingency tables that involve a large number of variables are sparse and their non-zero counts are imbalanced. That is, almost all the counts in large tables are zero; most of their positive counts are small (1, 2 or 3), and there are always a few counts that are quite large. Sparsity and imbalance give rise to severe difficulties when performing model selection due to the invalidation of the asymptotic approximations to the null distribution of the generalized likelihood-ratio test statistic, or the non-existence of the maximum likelihood estimates \cite{fienbergrinaldo2007,fienberg-rinaldo-2012}.

The Bayesian paradigm avoids some of these issues through the specification of prior distributions for model parameters \cite{clydegeorge2004}. \cite{dellaportasforster1999} represents a key contribution that proposed a Markov chain Monte Carlo (MCMC) algorithm to identify loglinear models with high posterior probability. Other notable papers develop various stochastic search schemes for discrete data \cite{madiganraftery1994,madiganyork1995,madiganyork1997,tarantola2004,dellaportastarantola2005,dobra-massam-2010}. These methods are known to work well for datasets with no more than $8$ variables. Another approach for Bayesian model selection in contingency tables is called copula Gaussian graphical models \cite{dobra-lenkoski-2011}, and it has successfully been used to analyze a 16-dimensional table. More recently, ultra-sparse high-dimensional contingency tables have been analyzed using probabilistic tensor factorizations induced through a Dirichlet process (DP) mixture model of product multinomial distributions \cite{dunson-xing-2009,canale-dunson-2011,bhattacharya-dunson-2012,kunihama-dunson-2013}. These papers present simulation studies and real-world data examples that involve up to $50$ categorical variables. 

\indent Penalized likelihood methods for categorical data have focused on Markov random fields for binary variables \cite{hoefling-tibshirani-2009,ravikumar-et-2010}. \cite{wainwright-jordan-2008} show that higher-order interactions and variables with three or more categories can be modeled by introducing additional binary variables in the model specification. Such claims have never been tested on known examples; from a theoretical perspective, there is no proof that the extension of the work of \cite{hoefling-tibshirani-2009} or \cite{ravikumar-et-2010} to general multi-way tables preserves the hierarchical structure of loglinear parameters, or yields consistent parameter and model estimates. The group lasso estimator for loglinear models \cite{nardi-rinaldo-2012}, despite having desirable theoretical properties, does not provide guarantees that the hierarchical structure of interaction terms is preserved.

In this paper we introduce a Bayesian framework for loglinear model determination that is suitable for the analysis of a contingency table with $214$ variables. Our method determines graphical loglinear models that are a special type of hierarchical loglinear models. Our key application comes from human mobility. In this context, multivariate categorical data capture the movement of individuals across multiple geographical areas. In Section \ref{sec:mobility} we discuss the relevance of massive unsolicited geolocated data for human mobility research, and in Section \ref{sec:mobilitymodels} we explain the role of loglinear models in modeling human movement. In Section \ref{sec:twitterdata} we describe our collection process of a geolocated Twitter dataset from South Africa; these data are subsequently transformed in the $214$ dimensional contingency table we analyze in Section \ref{sec:twitteranalysis}. Our modeling framework is presented in Section \ref{sec:intro GMs}. In Section \ref{sec:simulation} we provide information about the efficiency of our proposed method in a simulation study. In Section \ref{sec:conclusions} we give some concluding comments.

\section{Research on human mobility}
\label{sec:mobility}

Human mobility, or movement over short or long distances for short or long periods of time, is an important yet under-studied phenomenon in the social and demographic sciences. Migration processes represent a special case of human mobility that involve movements over longer periods of time and over longer distances. The impact of migration on human well-being, macro-social, political, and economic organization is a hot topic in the current literature \cite{Donato-1993,Durand-et-al-1996,Harris-Todaro-1970,Massey-1990,Massey-et-al-1993,Massey-Espinosa-1997,Massey-et-al-2010,Stark-Bloom-1985,Stark-Taylor-1991,Taylor-1987,Todaro-1969,Todaro-Maruszko-1987,VanWey-2005,Williams-2009}. Similar advances in understanding human mobility have been hindered by difficulties in recording and measuring how humans move on a minute and detailed scale. A notable exception is the relatively rich literature focusing on urban mobility and transportation studies. But much of this literature relies on travel surveys which are expensive to collect, have small sample sizes and limited spatial and temporal scales, are updated infrequently, and suffer from recall bias \cite{calabrese-et-al-2013,stopher-greaves-2007,wolf-et-2003}. Until recently, studies of mobility could not benefit from large scale data to widely address how differentials in mobility influence other outcomes. This is quite problematic given that mobility is likely a fundamental factor in behavior and macro-level social change, with potential associations with key issues that face human societies today, including spread of infectious diseases, responses to armed conflict and natural disasters, health behaviors and outcomes, economic, social, and political well-being, and migration.

Massive unsolicited geolocated data from mobile phones have recently become available for the study of human mobility. Such data are continuously collected by social media websites, search engines and wireless-service providers \cite{becker-et-al-2013}. Every time a person makes a voice call, sends a text message, goes online or makes posts through a social media service from their mobile phone, a record is generated with information about the time and day, duration and type of communication, as well as positional information. This could be the exact latitude and longitude of the mobile phone, or an identifier of the cellular tower that handled the request. The approximate spatiotemporal trajectory of a mobile phone and its user can be reconstructed by linking the records associated with that phone. This exciting new type of data holds immense promise for studying human behavior with precision and accuracy on a vast scale never before possible with surveys or other data collection techniques \cite{tatem-2014,dobra-et-2015,williams-et-2015}.  

User communications and check-ins through social media platforms such as Twitter generate publicly-available world-wide databases of human activity that can be readily accessed online free of charge. Recent evidence suggests that Twitter is a reliable source for examining human mobility patterns whose quality is comparable at the ecological level with mobile phone call records \cite{10.1371/journal.pone.0131469}. The cultural role of Twitter which serves a dual role as both a microblog and a social network, is evidenced by the Library of Congress' decision to store a permanent, daily updated archive of the site from its first tweet. Social media offers location sharing services whose growing popularity generate digital traces that can be located in space and time. Each day, Twitter records $7$ million tweets with explicit geolocation (latitude and longitude) information from mobile devices with GPS sensors \cite{Neubauer2015} that represent about 1.6\% of the total number of tweets \cite{leetaru-et-2013}.  The geographic information from geolocated tweets (geotweets) reveals the locations of human settlements and transportation networks \cite{leetaru-et-2013}. As the number of smartphone users continues to rise around the world, especially in low income countries, the potential of geolocated social media data to improve our knowledge of human geography will constantly grow. These are the data we collect and analyze in this paper \--- see Sections \ref{sec:twitterdata} and \ref{sec:twitteranalysis}.

\section{Modeling human mobility}
\label{sec:mobilitymodels}

The majority of the literature on human mobility is concerned with L\`{e}vy flights models and with Markov process models. Let us assume that traveling patterns are observed with respect to $p$ distinct areas or locations $\{1,2,\ldots,p\}$. Denote by $N_{ij}$ the number of individuals that traveled from location $i$ to location $j$ in a given time interval, and by $\Pr_{ij}$ the probability that a random individual will travel to location $j$ given that they are currently at location $i$. A class of stochastic process models called L\`{e}vy flights \cite{Brockmann2006} is one of the most popular way of modeling human mobility, or to model its limits \cite{Gonzalez2008}. This model represents the probability of traveling a distance $d$ as a power law: $\Pr(d) \propto d^{-(1+\beta)}$, where $\beta<2$ is a diffusion parameter. The L\`{e}vy flight model says that traveling a shorter distance is more likely than traveling longer distances, but long-distance travel can still occur even if it is rare. While this assumption is reasonable, the model implies that $\Pr_{ij}$ depends exclusively on $d_{ij}$ \--- the distance between locations $i$ and $j$. This represents a serious limitation since it implies that traveling to destinations that are located at the same distance from an origin is equally likely. A more recent contribution \cite{Guerzhoy2014} builds on multiplicative factor models from social network analysis \cite{Hoff2008} to improve the L\`{e}vy flights model which lacks the ability to quantify the desirability of certain travel locations. They propose a model in which $\Pr_{ij}$ depends of a function $f(d_{ij},\tau)$ of distance $d_{ij}$ and of location-specific latent factors $\mathbf{u}_i\in \reals^q$ and $\mathbf{v}_j\in \reals^q$: $\Pr_{ij} \propto \exp\left( f(d_{ij},\tau) + \mathbf{u}_i^T\mathbf{v}_j\right)$, where $\mathbf{u}_i^T\mathbf{v}_j$ represent the affinity of locations $i$ and $j$. Inference for this latent factor model is performed based on its log-likelihood that is proportional to $-\sum_{i,j} N_{ij}\log P_{ij}$.

Both the L\`{e}vy flights models \cite{Brockmann2006} and the multiplicative latent factor models \cite{Guerzhoy2014} are based on the crude assumption that human travel can be seen as a Markov process in which the probability of traveling to a location depends only on the origin of the trip's segment, and does not depend on previous locations visited. However, individuals are likely to travel repeatedly across multiple locations in a given period of time. Markov process models break mobility trajectories that involve multiple locations into pairs of consecutive locations, and, by doing so, loose key dependencies that are induced by multiple locations being visited by the same individuals in the reference time frame.

Loglinear models also have a long tradition in the human mobility literature, specifically, to estimate flows of migration by origin, destination, age, sex and other categorical sociodemographic variables such as economic activity group \cite{raymer-et-2007,smith-et-2010,raymer-et-2013}. Migration flows are represented as origin-destination migration flow tables. These are square tables in which the rows and columns correspond with places, regions, aggregation of places or countries of interest. The $(i,j)$ cell contains a count of the number of individuals that left from region $i$ and moved to region $j$ over the course of a specified time frame. The inclusion of other categorical variables lead to higher-dimensional migration flow tables. Modeling these tables involves spatial interaction loglinear models of the form \cite{raymer-et-2007}:
$$
 \log(\lambda_{ijk}) = \log(\alpha_i) + \log(\beta_j) + \log(m_{ijk}),
$$
where $\lambda_{ijk}$ is the expected migration flow from origin $i$ to destination $j$ for a combination of levels $k$ of one, two or more additional categorical variables, and $m_{ijk}$ is auxiliary information on the migration flow. The characteristics of the origin $i$ and the destination $j$ are represented through the parameters $\alpha_{i}$ and $\beta_{j}$. However, migration flow tables cannot capture the movement of those individuals that live in more than three regions during the time frame of observation. An example individual that left from region $1$ to move to region $2$, then moved again to region $3$, would contribute with a count of $1$ in the $(1,2)$ and $(2,3)$ cells of the resulting migration flow table. But, the link between these two counts will be lost. For this reason, loglinear models that estimate migration flows suffer from the same shortcoming as Markov process models.

\section{Description of the geolocated Twitter data}
\label{sec:twitterdata}

In this article we analyze a large-scale database of geolocated tweets from South Africa. This sub-Saharan country has been selected due to its high rates of internal and external migration caused by violent internal conflicts, war, political and economical instability, poverty, racial discrimination. Statistics South Africa reports that, in October $2016$, $3.5$ million travelers passed through South Africa's ports of entry. They were made up of $925,796$ South African residents and $2.6$ million foreign travelers. In this country, human mobility is known to be one of the major contributors to the spread of infectious diseases (HIV, tuberculosis, malaria) \cite{tatem-2014,dobra-et-2017}. 

Our geotweets database was put together in a two step process. First, geolocated tweets posted in South Africa between September $2011$ and September 2016 have been obtained directly from Twitter through GNIP, a reseller of social data owned by Twitter, as part of a no-cost collaborative research agreement between the University of Washington and Twitter.  A geotweet is classified to have been posted inside a country based on a country code field derived by GNIP from the latitude and longitude of the tweet. Second, we used the Twitter REST APIs \cite{TwitterRESTAPIs} to obtain geolocated tweets of the $476,601$ users whose geotweets have been captured in the first step. The REST APIs allow access to up to $3,200$ most recent geotweets in each user's timeline irrespective of the time when they have been posted, or the location they have been posted from. For this purpose, we used a customized version of the {\tt smappR} R package \cite{smappR}.  The second data collection step took place continuously between January and December 2016. During this period, the most recent geotweets of each of the $476,601$ users have been retrieved at least twice per month.

The total number of unique geotweets acquired in both steps is $46,210,370$. The actual tweets have been discarded after we extracted tuples of the form {\small{\tt <user key, time of the posting, latitude, longitude, ...>}} from the rich content of each tweet. To assure privacy protection, each Twitter user is identified by a randomly generated key which replaces their Twitter identifier. Additional filtering steps were performed to eliminate any non-human activity (e.g., Twitter robots) or any geotweets with coding errors. We emphasize that this database comprises only public information which can be viewed online, and replicated using the APIs provided by Twitter or downloaded directly from a third party provider of social media data such as GNIP. 

For each of the $476,601$ users, we determined their country of residence. We estimated the amount of time a user spent in a country they visited as the cumulative periods of time between consecutive geotweets posted in that country. A user's country of residence was defined as the country with the largest amount of time spent among all the countries this user tweeted from. Our method for identifying the users' country of residence has certain limitations. First, it is possible that a user could choose to post geotweets only when they are away from their country of residence. Second, it is also possible that our two step process of collecting geotweets might have missed relevant time intervals in which a user tweeted from their country of residence. However, after carefully examining the spatial patterns of geotweets with respect to the estimated countries of residence, we are confident that our method of determination worked fine for a large percentage of users. Based on this information, we classified $41,049$ ($8.62$\%) of the $476,601$ users as visitors of South Africa, and the rest as locals, that is, individuals that most likely see South Africa as their home country. 

\begin{figure} [ht] 
  \begin{center}
    \includegraphics[width=9cm]{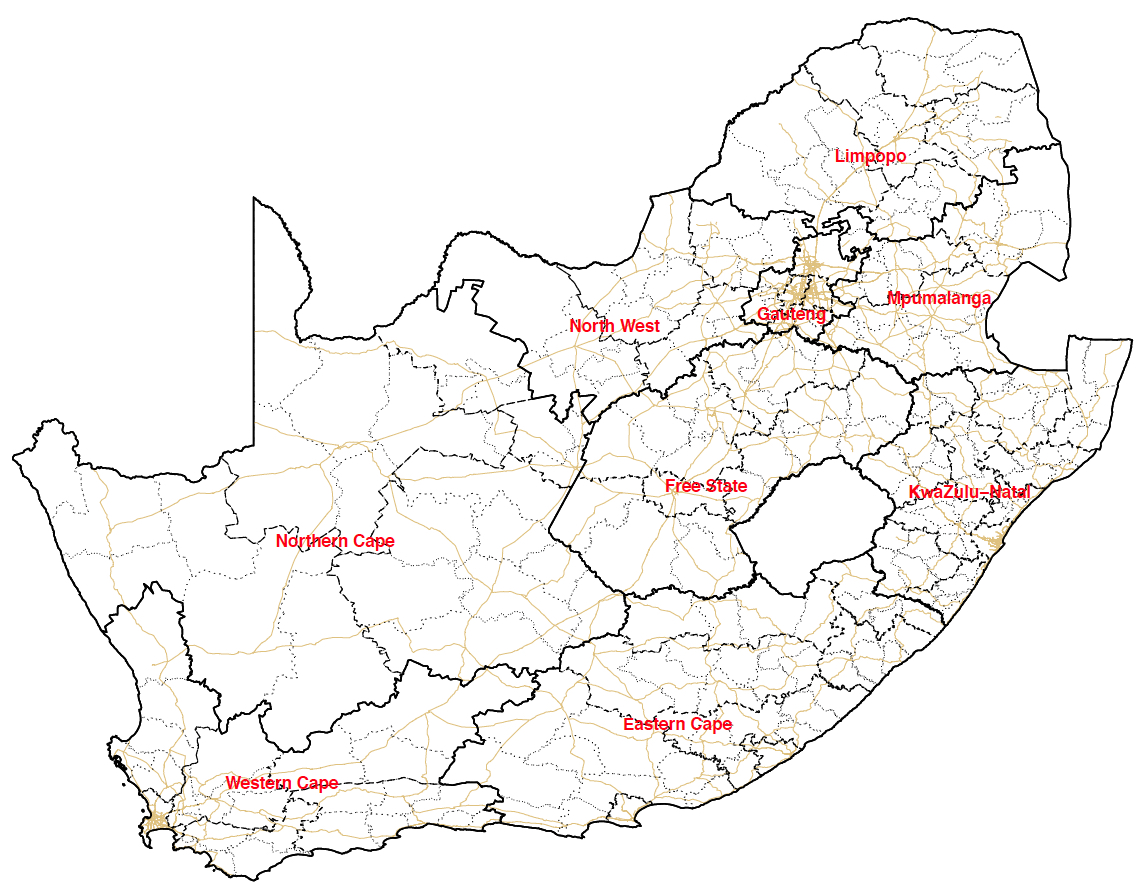}
  \end{center}
  \caption{\label{fig:rsaMap} Administrative divisions of South Africa: nine provinces divided into $52$ metropolitan and district municipalities (dashed lines) that are further divided into $213$ local municipalities (dashed lines). The motorways, trunk, and primary roads (source: \cite{OpenStreetMap}) of the country are also shown.}
\end{figure}

We subsequently mapped the geotweets into the $213$ district municipalities of South Africa \--- see Figure \ref{fig:rsaMap}. This allowed us to determine, for each user, the municipalities they were present and absent during the five years data collection time frame. Here we assume that absence from a municipality is implied by the user not posting any geotweets within its boundaries. These presence/absence patterns together with the Local (yes/no) variable define a $214$ dimensional binary contingency table. This table is hyper-sparse: only $55015$ cells contain positive counts (the logarithm of the percentage of non-zero counts is $-132.813$). Among the $55015$ non-zero counts, there are $46175$ ($83.93$\%) counts of $1$, $3439$ ($6.25$\%) counts of $2$, $1411$ ($2.56$\%) counts of $3$, $747$ ($1.36$\%) counts of $4$, and $476$ ($0.87$\%) counts of $5$. The top five largest counts are $58929$, $42781$, $28731$, $28197$ and $22313$, and represent the number of users that were locals to South Africa and posted geotweets only from one of following five metropolitan municipalities: Johannesburg (JHB, Gauteng), Cape Town (CPT, Western Cape), Tshwane (TSH, Gauteng), eThekwini (ETH, KwaZulu-Natal), and Ekurhuleni (EKU, Gauteng), respectively. The sixth largest count is 9568, and represents the number of users that were locals to South Africa, and posted geotweets from two metropolitan municipalities, Johannesburg (JHB) and Ekurhuleni (EKU). The seventh largest count count is 8464, and represents the number of users that were visitors (non-locals) to South Africa, and posted geotweets only from Johannesburg (JHB). In the next section we present our framework for determining the multivariate patterns of interactions among these 214 binary variables.

\section{Bayesian structural learning in graphical loglinear models}
\label{sec:intro GMs}

An undirected interaction graph $G=(V,E)$ ($V = \{1,...,p\}$ are vertices, and  $E \subset V\times V$ are edges) is defined for a hierarchical loglinear model $\mathcal{H}$ that involves $p$ categorical variables $\bfX = (X_1,X_2,\ldots,X_p)$ as follows. A vertex $i\in V$ of $G$ corresponds with variable $X_i$. An edge $e=(i,j)$ appears in $G$ if and only if the variables $X_i$ and $X_j$ appear together in an interaction term of $\mathcal{H}$. Model $\mathcal{H}$ is graphical if the subsets of $V$ that are the vertices of the complete subgraphs of $G$ that are maximal with respect to inclusion, are also maximal interaction terms in $\mathcal{H}$ \cite{lauritzen1996graphical}. In this case, the absence of an edge between vertices $i$ and $j$ in $G$ means that $X_i$ and $X_j$ are conditional independent given the remaining variables $X_{V\setminus \{i,j\}}$. For this reason, the interaction graph $G$ of a graphical loglinear model is called a conditional independence graph. This graph also has a predictive interpretation. Denote by $\nbd_G(i) = \{ j\in V: (i,j)\in E\}$ the neighbors of vertex $i$ in $G$. Then $X_i$ is conditionally independent of $X_{V\setminus (\nbd_G(i)\cup \{i\})}$ given $X_{\nbd_G(i)}$ which implies that, given $G$, a mean squared optimal prediction of $X_i$ can be made from the neighboring variables $X_{\nbd_G(i)}$.

We focus on the structural learning problem \cite{jones2005experiments,drton-maathuis-2017} which aims to estimate the structure of $G$ (i.e., which edges are present or absent in $E$) from the available data $\xb = (x^{(1)},..., x^{(n)})$. In a Bayesian framework, we explore the posterior distribution of $G$ conditional on the data $\xb$, i.e.
\begin{equation}
 \label{posterior dis}
 \Pr( G \mid \xb ) = \frac{\Pr(G) \Pr(\xb\mid G)} {\sum_{G \in \mathcal{G}_p} \Pr(G) \Pr(\xb\mid G)},
\end{equation}
\noindent where $\Pr(G)$ is a prior distribution on the space $\mathcal{G}_p$ of undirected graphs with $p$ vertices, and $\Pr(\xb\mid G)$ is the marginal likelihood of the data conditional on $G$ \cite{jones2005experiments}. Identifying the graphs with the largest posterior probability \eqref{posterior dis} is a complex problem because the number of undirected graphs $2^{{p \choose 2}}$ in $\mathcal{G}_p$ becomes large very fast as $p$ increases. For example, for $p=20$, the number of undirected graphs in $\mathcal{G}_p$ exceeds $10^{70}$. In this paper we introduce a computationally efficient search algorithm that takes advantage of parallelizable local computations at the vertex level that moves fast towards regions with high posterior  probabilities \eqref{posterior dis}.

\subsection{Bayesian structural learning via birth-death processes}
\label{subsec:graph selection}

To efficiently explore the graph space $\mathcal{G}_p$, \cite{mohammadi2015bayesianStructure} developed the birth-death Markov chain Monte Carlo (BDMCMC) algorithm. This is a trans-dimensional MCMC algorithm, and represents an alternative to the well known reversible jump MCMC algorithm \cite{green1995reversible}. The version of BDMCMC presented in \cite{mohammadi2015bayesianStructure} was developed specifically for Gaussian graphical models. In this section we give a general formulation for sampling from any distributions on a space of graphs $\mathcal{G}_p$.

The BDMCMC algorithm is based on a continuous time birth-death Markov process \cite{preston1976}. Its underlying sampling scheme traverses $\mathcal{G}_p$ by adding and removing edges corresponding to the birth and death events. Given that the process is at state $G=(V,E)$, we define the birth and death events as independent Poisson processes as follows:

\textit{Birth event} \--- each edge $e \in \Ehat$ where $\Ehat = \{ e\in V\times V : e \notin E \}$, is born independently of other edges that do not belong to $G$ as a Poisson process with rate $B_e(G)$. If the birth of edge $e$ occurs, the process jumps to $G^{+e}=(V,E \cup \{e\} )$ which is a graph with one edge more than $G$.

\textit{Death event} \---  each edge $e \in E$ dies independently of other edges that belong to $G$  as a Poisson process with rate $D_e(G)$. If the death of edge $e$ occurs, the process jumps to $G^{-e}=(V,E \setminus \{e\})$ which is a graph with one edge less than $G$.

This birth-death Markov process is a jump process with intensity $\alpha(G) = \sum_{e \in \Ehat}{B_{e}(G)}+\sum_{e \in E}{D_{e}(G)}$. Its waiting time to the next jump follows an exponential distribution with expectation $1/\alpha(G)$. The birth and death probabilities are 
\begin{eqnarray}
 \label{prob.birth}
 \Pr( \mbox{birth of edge } e ) & \varpropto & B_e(G), \quad \mbox{for } e \in \Ehat,\\
 \label{prob.death}
 \Pr( \mbox{death of edge } e ) & \varpropto & D_e(G), \quad \mbox{for } e \in E.
\end{eqnarray}
The following theorem provides sufficient conditions on the birth and death rates to guarantee that the corresponding process on $\mathcal{G}_p$ has stationary distribution (\ref{posterior dis}).
\begin{theorem}
\label{theorem:bd}
 The birth-death process defined by the birth and death probabilities (\ref{prob.birth}) and (\ref{prob.death}) has the stationary distribution $\Pr(G\mid \xb)$  given in (\ref{posterior dis}), if the following detailed balance condition is satisfied:
  \begin{equation}
  \label{eq:bd}
    B_{e}(G) \Pr(G \mid \xb) = D_{e}(G^{+e}) \Pr(G^{+e} \mid \xb),
  \end{equation}
where $e \in \Ehat$, $G=(V,E)$, and $G^{+e}=(V,E \cup \{e\} )$.
\end{theorem}

\begin{proof}
We take advantage of the theory on general classes of Markov birth-death processes from \cite[Section 7 and 8]{preston1976}. This class of Markov jump processes evolve in jumps which occur a finite number of times in any finite time interval. These jumps are of two types: (i) \textit{birth} in which a single point is added, and the process jumps to a state that contains the additional point; and (ii) \textit{death} in which one of the points in the current state is deleted, and the process jumps to a state with one less point. \cite{preston1976} shows that the process converges to a unique stationary distribution provided that the detailed balance conditions hold. 

To define the balance conditions for our process, assume that at a given time, the process is in a graph state $G=(V,E)$ with $\theta_G \in \Theta_G$ as a vector of parameters. The process is characterized by the {\it birth rates} $B_e(G, \theta_G)$ for each $e \in \Ehat$, the {\it death rates} $D_e(G, \theta_G)$ for each $e \in E$, and the birth and death \textit{transition kernels} $K^{(G)}_{B_e}(\theta_G;\cdot)$ and $K^{(G)}_{D_e}(\theta_G;\cdot)$. Birth and death events occur as independent Poisson processes with rates $B_e(G, \theta_G)$ and $D_e(G,\theta_G)$ respectively. Given the birth of $e \in \Ehat$ occurs, the probability that the following jump leads to a point in $\F \in \Theta_{G^{+e}}$ is 
\begin{eqnarray*}
K^{(G)}_{B_e}(\theta_G; \F) = \frac{B_{e}(G, \theta_G)}{B(G, \theta_G)} \int_{ \theta_e : \theta_G \cup \theta_e \in \F}{b_e(\theta_e; \theta_G) \di \theta_e},
\end{eqnarray*}
in which $B(G, \theta_G)=\sum_{e \in \Ehat} B_e(G, \theta_G)$. Similarly, given the death of $e \in E$ occurs, the probability that the following jump leads to a point in $\F \in \Theta_{G^{-e}}$ is 
\begin{eqnarray}
\label{kernel 2}
K^{(G)}_{D_e}(\theta_G; \F) = \frac{D_{e}(G, \theta_G)}{D(G, \theta_G)} \I(\theta_{G^{-e}} \in \F),
\end{eqnarray}
in which $D(G, \theta_G)=\sum_{e \in E} D_e(G, \theta_G)$. 

For this birth-death process, $\Pr(G, \theta_G\mid \xb)$ satisfies detailed balance conditions if
\begin{eqnarray}
\label{balance1}
\int_{\F} B(G,\theta_G) \Pr(G, \theta_G\mid \xb) \di \theta_{G} = \\ \nonumber
\sum_{e \in \Ehat} \int_{\Theta_{G^{+e}}} D(G^{+e},\theta_{G^{+e}}) K^{(G^{+e})}_{D_e}(\theta_{G^{+e}}; \F) \Pr(G^{+e},\theta_{G^{+e}} \mid \xb) \di \theta_{G^{+e}},
\end{eqnarray}
and
\begin{eqnarray}
\label{balance2}
\int_{\F} D(G,\theta_G) \Pr(G, \theta_G \mid \xb) \di \theta_{G} = \\  \nonumber
\sum_{e \in E} \int_{\Theta_{G^{-e}}} B(G^{-e}, \theta_{G^{-e}}) K^{(G^{-e})}_{B_e}(\theta_{G^{-e}}; \F) \Pr(G^{-e}, \theta_{G^{-e}} \mid \xb) \di \theta_{G^{-e}},
\end{eqnarray}
where $\F \subset \Theta_{G}$. 

We check the first part of the detailed balance conditions \eqref{balance1} as follows
\begin{align*}
LHS &= \int_{\F} B(G, \theta_G) \Pr(G,\theta_G\mid \xb) \di \theta_{G} \\
    &= \int_{\Theta_G} \I(\theta_G \in \F) B(G, \theta_G) \Pr(G,\theta_G\mid \xb) \di \theta_{G} \\
    &= \int_{\Theta_G} \I(\theta_G \in \F) \sum_{e \in \Ehat}{B_{e}(G,\theta_G)} \Pr(G,\theta_G\mid \xb) \di \theta_{G} \\
    &= \sum_{e \in \Ehat} \int_{\Theta_G} \I(\theta_G \in \F) B_{e}(G,\theta_G) \Pr(G,\theta_G\mid \xb) \di \theta_{G} \\
    &= \sum_{e \in \Ehat} \int_{\Theta_G} \I(\theta_G \in \F) B_{e}(G,\theta_G) \Pr(G,\theta_G\mid \xb) \left[ \int_{\Theta_e} b_e(\theta_e; \theta_G) \di\theta_e \right] \di \theta_G \\
    & \qquad \qquad \qquad \qquad \qquad \qquad  \qquad \qquad \qquad \qquad \left[ b_e \text{ must integrate to }1 \right] \\
    &= \sum_{e \in \Ehat} \int_{\Theta_G} \int_{\Theta_e} \I(\theta_G \in \F) B_{e}(G,\theta_G) \Pr(G,\theta_G\mid \xb) b_e(\theta_e; \theta_G) \di \theta_e \di \theta_G. \\
\end{align*}
\begin{align*}
RHS &= \sum_{e \in \Ehat}{ \int_{\Theta_{G^{+e}}} D(G^{+e},\theta_{G^{+e}}) K^{(G^{+e})}_{D_e}(\theta_{G^{+e}}; \F) \Pr(G^{+e},\theta_{G^{+e}} \mid \xb) \di \theta_{G^{+e}} } \\
    &  \qquad \qquad \qquad \qquad \qquad \qquad  \qquad \qquad \qquad \qquad \qquad \left[ \text{equation \eqref{kernel 2} } \right] \\
    &= \sum_{e \in \Ehat}{ \int_{\Theta_{G^{+e}}} \I(\theta_G \in \F) D_{e}(G^{+e}, \theta_{G^{+e}}) \Pr(G^{+e},\theta_{G^{+e}} \mid \xb) \di \theta_{G^{+e}} }.
\end{align*}
Therefore we have LHS=RHS provided that
\begin{eqnarray*}
B_{e}(G,\theta_G) \Pr(G,\theta_G\mid \xb) b_e(\theta_e ; \theta_G) = D_{e}(G^{+e}, \theta_{G^{+e}}) \Pr(G^{+e},\theta_{G^{+e}} \mid \xb).
\end{eqnarray*} 
Now, by integrating over $\theta_{G^{+e}} = \theta_{G} \cup \theta_e$ and knowing that the function $b_e(\theta_e; \theta_{G})$ must integrate to $1$ over $\Theta_e$, we have 
\begin{eqnarray*}
B_{e}(G) \Pr(G \mid \xb) = D_{e}(G^{+e}) \Pr(G^{+e} \mid \xb),
\end{eqnarray*}
which is the expression (\ref{eq:bd}) in Theorem \ref{theorem:bd}. In a similar way, it can be shown that the remaining part of the detailed balance conditions \eqref{balance2} also hold.
\end{proof} 

Based on Theorem \ref{theorem:bd}, we define the birth and death rates of the BDMCMC algorithm as a function of the ratio of the corresponding posterior probabilities to optimize the convergence speed:
\begin{equation*}
  \label{birthrate}
  B_e(G) = \min \left\{ \frac{\Pr(G^{+e} \mid \xb)}{\Pr(G \mid \xb)}, 1 \right\}, \ \ \mbox{for each} \ \ e \in \Ehat,
\end{equation*}
\begin{equation*}
  \label{deathrate}
  D_e(G) = \min \left\{ \frac{\Pr(G^{-e} \mid \xb)}{\Pr(G \mid \xb )}, 1 \right\}, \ \ \mbox{for each} \ \ e \in E.
\end{equation*}
We show the birth and death rates as follows
\begin{equation}
  \label{rate}
  R_e(G) = \min \left\{ \frac{\Pr(G^{*} \mid \xb)}{\Pr(G \mid \xb)}, 1 \right\}, \ \ \mbox{for each} \ \ e \in \{ E \cup \Ehat \},
\end{equation}
where for the birth of edge $e$ we take $G^*=(V,E \cup \{e\} )$, and for the death of edge $e$ we take $G^*=(V,E \setminus \{e\} )$.

Algorithm \ref{algorithm:BDMCMC} provides the pseudo-code for the BDMCMC algorithm which samples from the posterior distribution \eqref{posterior dis} on $\mathcal{G}_p$ by using the above birth-death mechanism. In Section \ref{subsec:bdmcmc} we explain how to efficiently compute the ratio of posterior probabilities in the birth and death rates \eqref{rate} for multivariate discrete data by using the marginal  pseudo-likelihood approach \cite{pensar2016marginal}.

\begin{algorithm}
\renewcommand{\algorithmicrequire}{\textbf{Input:}}
\renewcommand{\algorithmicensure}{\textbf{Output:}} 
\caption{. BDMCMC algorithm for undirected graphical models}
\label{algorithm:BDMCMC}
\begin{algorithmic}[0]
\REQUIRE A graph $G=(V,E)$ with $p$ nodes and data $\xb$
\FOR {$N$ iterations}
  \FOR {all the possible edges in parallel}
    \STATE Calculate the birth and death rates in (\ref{rate}),
  \ENDFOR
  \STATE Calculate the waiting time for $G$ by $W(G)= \frac{1}{\sum_{e \in \Ehat}{B_{e}(G)}+\sum_{e \in E}{D_{e}(G)}} $
  \STATE Update $G$ based on birth/death probabilities in (\ref{prob.birth}) and (\ref{prob.death})
\ENDFOR
\ENSURE Samples from the posterior distribution \eqref{posterior dis}. 
\end{algorithmic}
\end{algorithm}
\begin{figure} [!ht] 
\centering
\includegraphics[width=1\textwidth]{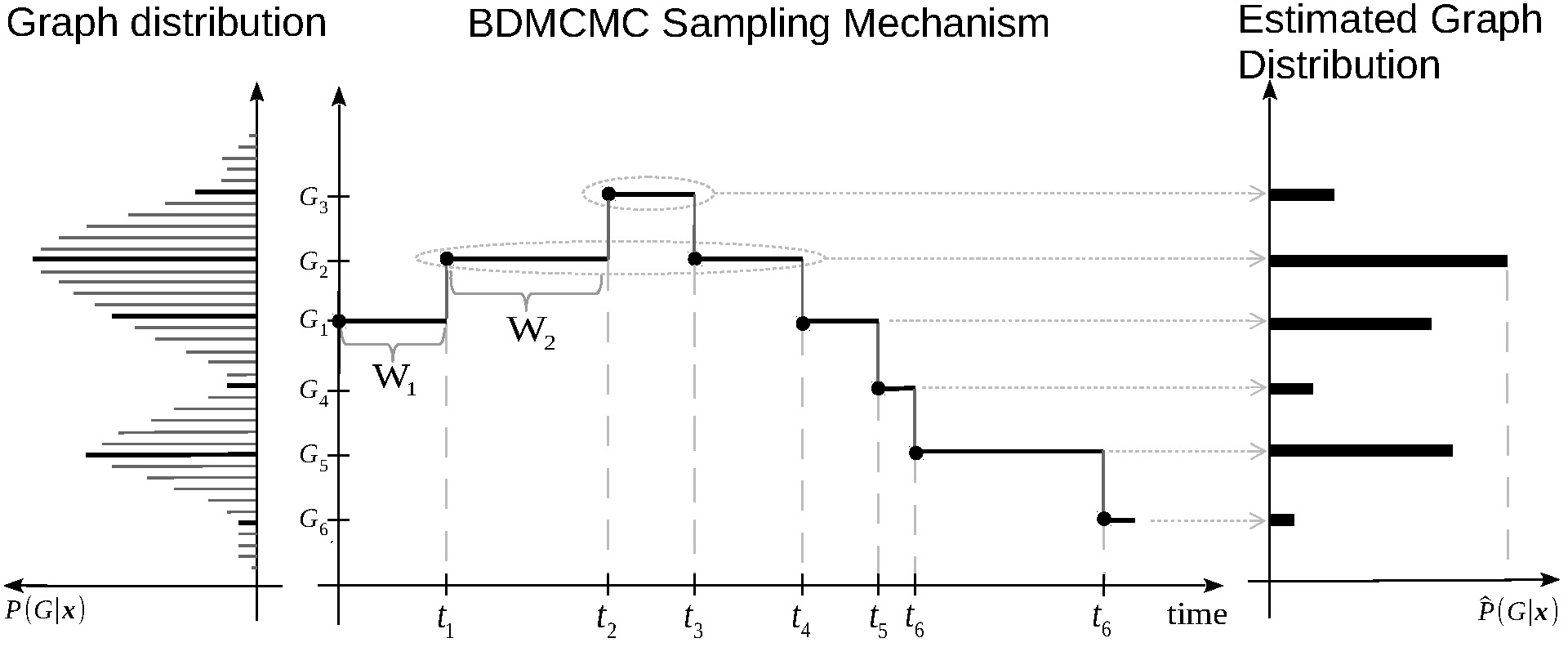}
\caption{ \label{fig:BDMCMC} 
The left and right panels show the true and estimated posterior distribution \eqref{posterior dis} on the space the graphs. The middle panel shows an example output from an application of Algorithm \ref{algorithm:BDMCMC} where $ \left\{ W_1, W_2,... \right\}$ denote 
waiting times, and $ \left\{ t_1, t_2,... \right\}$ denote jumping times.}
\end{figure}

\subsection{Posterior estimation via sampling in continuous time}
\label{subsec:estimation}

Figure \ref{algorithm:BDMCMC} illustrates how the output of Algorithm \ref{algorithm:BDMCMC} can be used to estimate posterior quantities of interest. The output consists of a set of sampled graphs, a set of waiting times $\left\{ W_1, W_2,... \right\}$, and a set of jumping times $\left\{ t_1, t_2,... \right\}$. Based on the Rao-Blackwellized estimator \cite{cappe2003reversible}, the estimated posterior probability of each sampled graph is proportional to the expectation of length of the holding time in that graph which is estimated as the sum of the waiting times in that graph. The posterior inclusion probability of an edge $e\in V\times V$ is estimated by
\begin{eqnarray}
\label{posterior-edge}
\widehat{\Pr}(\mbox{edge } e \mid \xb )= \frac{\sum_{t=1}^{N}{\I(e \in G^{(t)}) W(G^{(t)}) }}{\sum_{t=1}^{N}{W(G^{(t)})}},
\end{eqnarray}
where $N$ denotes the number of iterations, $\I(e \in G^{(t)})$ denotes an indicator function: $\I(e \in G^{(t)})=1$ if $e \in G^{(t)}$, and 
$0$ otherwise. 

\subsection{Birth and death rates with the marginal pseudo-likelihood}
\label{subsec:bdmcmc}

We assume that the observed random variables $\bfX = (X_1,X_2,\ldots,X_p)$ are categorical, with each variable $X_i$ taking values in a discrete set $\mathcal{X}_i=\{1, 2,...,r_i\}$. The determination of the birth and death rates \eqref{rate} involves the marginal likelihood conditional on a graph $G\in \mathcal{G}_p$:
\begin{equation}
\label{eq:pseudo-likelihood}
 \Pr( \xb \mid G ) = \int_{\Theta_G} \Pr( \xb \mid \theta_G, G ) \Pr( \theta_G \mid G ) \di \theta_G,
\end{equation}
where $\theta_G \in \Theta_G$ are the parameters of a multivariate model associated with $G$, $\Pr( \theta_G \mid G )$ is prior for $\theta_G$, and $\Pr(\xb\mid \theta_G, G )$ is the full likelihood function. However, the exact calculation of the marginal likelihood $\Pr( \xb\mid G )$ is possible only for decomposable graphs $G$ which represent a small fraction of the graphs in $\mathcal{G}_p$ \cite{massamliudobra2009}. Numerical approximations for the marginal likelihood for arbitrary undirected graphs have been developed \cite{dobra-massam-2010}, but their application is  computationally expensive even for datasets that involve $p<20$ variables. This high computational effort renders them inapplicable for the Twitter mobility data described in Section \ref{sec:twitterdata} with $p=214$ observed variables.

A computationally cheaper alternative comes from approximating the full likelihood $\Pr(\xb\mid \theta_G, G )$ with the pseudo-likelihood \cite{besag1975statistical, besag1977efficiency} which is the product of the full conditionals of the random variables $\bfX$ given their neighbors in $G$:
\begin{equation}
\label{eq:pseudolikfirst}
\Pr_{pl}( \xb\mid \theta^{pl}_G, G ) = \prod_{d=1}^n\prod_{i=1}^p  \Pr(X_i=x^{(d)}_i\mid \bfX_{\nbd_G(i)}=\bfx^{(d)}_{\nbd_G(i)}, \theta^{pl}_{i,G}).
\end{equation}
\noindent We denote  $\mathcal{X}_{\nbd_G(i)} = \bigtimes_{j\in \nbd_G(i)}\mathcal{X}_j$, $\theta_{i,+ l} = \{ \theta_{i,kl}:k\in \mathcal{X}_i\}$, and $\theta^{pl}_{i,G} = \{ \theta_{i,+ l}: l \in  \mathcal{X}_{\nbd_G(i)}\}$. In \eqref{eq:pseudolikfirst}, $\theta^{pl}_G = \bigtimes_{i=1}^p \theta^{pl}_{i,G}\in \Theta^{pl}_G$ are the set of parameters of the full conditionals
\begin{equation*}
 \Pr( X_i = k\mid \bfX_{\nbd_G(i)} = l ) = \theta_{i,kl},  \quad \textrm{for } i = 1,...,p,
\end{equation*}
where $k\in \mathcal{X}_i$, $l\in \mathcal{X}_{\nbd_G(i)}$. Thus, the pseudo-likelihood \eqref{eq:pseudolikfirst} can be written as:
\begin{equation}
\label{eq:pseudolik}
\Pr_{pl}( \xb\mid \theta^{pl}_G, G ) = \prod_{i=1}^p \prod_{k\in \mathcal{X}_i}\prod_{l\in \mathcal{X}_{\nbd_G(i)}} \theta_{i,kl}^{n_{i,kl}},
\end{equation}
where $n_{i,kl}$ represents the number of samples $x^{(d)}$, $d=1,2,\ldots,n$, such that $x^{(d)}_i = k$ and $x^{(d)}_{\nbd_G(i)}=l$.

For computational convenience, we assume that the set of parameters $\theta^{pl}_{i,G}$ and $\theta^{pl}_{i^{\prime},G}$ associated with the full conditionals of $X_i$ and $X_{i^{\prime}}$, $i\ne  i^{\prime}$ are independent. This assumption is certainly not consistent with the assumption that the full conditionals are derived from the same full joint distribution of $\bfX$. Nevertheless, the approximation of the full likelihood with the pseudo-likelihood \eqref{eq:pseudo-likelihood} is based on the same premise \cite{besag1975statistical, besag1977efficiency}. We also assume that, within the same full conditional associated with the variable $X_i$, the parameters $\theta_{i,+ l}$ and $\theta_{i,+ l^{\prime}}$ associated with the different levels $l$ and $l^{\prime}$ of the variables $\bfX_{\nbd_G(i)}$ are independent \cite{pensar2016marginal}. We impose a prior for $\theta^{pl}_G$ that factorizes according to these two assumptions:
\begin{equation}
\label{eq:pseudoprior}
\Pr( \theta^{pl}_G ) = \prod_{i=1}^p \Pr( \theta_{i,G} ) = \prod_{i=1}^p \prod_{l\in \mathcal{X}_{\nbd_G(i)}} \Pr( \theta_{i,+ l} ).
\end{equation}
Furthermore, we impose a Dirichlet prior on the conditional probabilities of $X_i$ at level $l$ of $\bfX_{\nbd_G(i)}$:
\begin{equation}
 \label{eq:dirichletprior}
 \theta_{i,+ l} \sim \Dir(\alpha_{i,1 l},\ldots,\alpha_{i,r_i l}).
\end{equation}
From \eqref{eq:pseudolik}, \eqref{eq:pseudoprior}, and \eqref{eq:dirichletprior}, it follows that the marginal pseudo-likelihood is \cite{pensar2016marginal}:
\begin{equation}
\label{eq:mpl}
 \Pr_{pl}( \xb \mid G ) = \prod_{i=1}^p \Pr( \bfx_i\mid \bfx_{\nbd_G(i)}),
\end{equation}
with
\begin{equation}
 \Pr( \bfx_i\mid \bfx_{\nbd_G(i)}) = \prod_{l\in \mathcal{X}_{\nbd_G(i)}} \frac{\Gamma \left( \alpha_{i,+l}\right) }{ \Gamma \left( \alpha_{i,+l} + n_{i,+l} \right) } 
\prod_{k\in \mathcal{X}_i} \frac{\Gamma \left( \alpha_{i,kl} + n_{i,kl} \right) }{ \Gamma \left( \alpha_{i,kl} \right) }, \label{eq:marklikreg}
\end{equation}
where $\alpha_{i,+l} = \sum_{k\in \mathcal{X}_i} \alpha_{i,kl}$ and $n_{i,+l} = \sum_{k\in \mathcal{X}_i} n_{i,kl}$.

A prior on the space of graphs $\mathcal{G}_p$ that encourages sparsity by penalizing for the inclusion of additional edges in the graph $G=(V,E)$ is \cite{jones2005experiments}:
\begin{equation}
 \label{eq:graphprior}
 \Pr(G) \propto \left(\frac{\beta}{1-\beta}\right)^{\mid E \mid} = \left(\prod_{i=1}^p \left(\frac{\beta}{1-\beta}\right)^{\mid \nbd_G(i) \mid}\right)^{1/2},
\end{equation}
where $\beta\in (0,1)$ is set to a small value, e.g. $\beta=1/{p \choose 2}$. While other priors on $\mathcal{G}_p$ are available \cite{dobra-et-2011}, the prior \eqref{eq:graphprior} can be decomposed as the product of independent priors for the $p$ full conditionals given $G$ such that the probability of inclusion of a vertex in each of these conditionals is equal with $\beta$ as shown in \eqref{eq:graphprior}.

The marginal posterior distribution on $\mathcal{G}_p$ based on the marginal pseudo-likelihood \eqref{eq:mpl} and the prior on $\mathcal{G}_p$ \eqref{eq:graphprior} is
\begin{equation}
\label{eq:postmpl}
 \Pr_{pl}( G \mid \xb ) \propto \Pr_{pl}( \xb \mid G ) \Pr(G) = \prod_{i=1}^p \Pr( \bfx_i\mid \bfx_{\nbd_G(i)}) \left(\frac{\beta}{1-\beta}\right)^{\frac{\mid \nbd_G(i) \mid}{2}}.
\end{equation}

The birth and death rates in \eqref{rate} based on the marginal pseudo-likelihood for an edge $e=(i,j)\in V\times V$ are calculated from
\begin{equation*}
 \label{rate-mpl}
 \widehat{R}_e(G) = \min \left\{ \frac{\Pr(\bfx_i \mid \bfx_{\nbd_{G^*}(i)})}{\Pr( \bfx_i\mid \bfx_{\nbd_{G}(i)})}  \frac{\Pr(\bfx_j \mid \bfx_{\nbd_{G^*}(j)})}{\Pr( \bfx_j\mid \bfx_{\nbd_{G}(j)})}  \left(\frac{\beta}{1-\beta}\right)^{\delta}, 1 \right\},
\end{equation*}
where for the birth of edge $e$ we take $G^*=(V,E \cup \{e\} )$, $\delta =1$, and for the death of edge $e$ we take $G^*=(V,E \setminus \{e\} )$, $\delta =-1$.

\subsection{Speeding up the BDMCMC algorithm}
\label{subsec:speedup}

The key bottleneck of the BDMCMC algorithm is the computation at every iteration of the birth and death rates \eqref{rate} for all the $p(p-1)/2$ possible edges. Fortunately, the rates associated with one edge can be calculated independently of the rates associated with the other edges, and can be performed in parallel which represents a first key computational improvement. We implemented parallel computations of the birth and death rates in the current version of the \texttt{R} package \texttt{BDgraph} \cite{BDgraph} using \texttt{OpenMP} \cite{openmp08}. Most code in this package is written in \texttt{C++} and interfaced in \texttt{R}.

A second key computational improvement is possible when the marginal likelihood is replaced with the marginal pseudo-likelihood as detailed in Section \ref{subsec:bdmcmc}. Since at each step of the BDMCMC algorithm one edge $e=(i,j)$ is selected for addition or removal, only the marginal likelihood  \eqref{eq:marklikreg} of the full conditionals of the two vertices $i$ and $j$ will change. Thus, we need to recalculate the $(p-1) + (p-1) - 1 = 2p-3$ rates that correspond with these two vertices. The remaining rates will stay the same. As such, at each iteration we update $2p-3$ rates instead of $p(p-1)/2$ rates. This represents a huge computational saving especially for graphs with many vertices. For example, for the Twitter mobility data we analyze in Section \ref{sec:twitteranalysis}, we look at graphs with $p=214$ vertices. Instead of computing $22791$ rates at each step of the BDMCMC algorithm, we only need to determine $422$ rates which means that a single edge update can be done approximately 54 times faster.

A third key computational improvement comes from allowing multiple edge updates at each iteration. The vast majority of the MCMC and stochastic search algorithms that have been developed in the Bayesian graphical models literature  are based on adding or removing one edge at each iteration \cite{jones2005experiments, lenkoski2011computational, scott2008feature, lenkoski2011computational, wang2012efficient, mohammadi2017bayesian, mohammadi2015bayesianStructure, mohammadi2017ratio, cheng2012hierarchical}. These single edge updates are in part responsible for making these structural learning algorithms quite slow for datasets that comprise a larger number of variables $p$. Multiple birth-death sampling approaches have been used to address image processing problems that aim to detect a configuration of objects from a digital image, and have been found to outperform the convergence speed of competing reversible jump MCMC algorithms \cite{descombes2009object, eldin2012multiple, gamal2010multiple, gamal2011fast}.

By following this idea, it is possible to transform Algorithm \ref{algorithm:BDMCMC} into a multiple birth-death MCMC algorithm based on a multiple birth-death process. At each iteration, after computing and ranking the birth and deaths rates \eqref{rate}, we select not one but a fixed number $N_0\ge 2$ of edges to be added or removed from the graph. By doing so, $N_0$ edges are updated at no computational cost compared to a single edge update. Through multiple edge updates which we have also implemented in the \texttt{R} package \texttt{BDgraph} \cite{BDgraph}, the BDMCMC algorithm can quickly move to regions with high posterior probability in the graph space $\mathcal{G}_p$. The ability to move towards high posterior probability graphs in a small number of iterations is especially important in applications in which the ratio between the number of samples available and the number of variables is small. However, performing multiple edge updates at each iteration of the BDMCMC algorithm does not have any theoretical guarantees related to sampling from the correct target posterior distribution \eqref{posterior dis}. Multiple edge updates in can be performed for a reduced number of iterations to identify several graphs that have higher posterior probabilities compared to the empty graph, the full graph or a random graph sampled from $\mathcal{G}_p$. These graphs can be subsequently used as starting points for Algorithm \ref{algorithm:BDMCMC} with single edge updates.

\section{Simulation Study}
\label{sec:simulation}

We investigate the performance of the BDMCMC algorithm in recovering the graph structure from categorical data by comparing it to the hill-climbing (HC) algorithm proposed by \cite{pensar2016marginal}. While the BDMCMC algorithm samples from the marginal posterior distribution \eqref{eq:postmpl}, the HC algorithm solves the optimization problem $\max\{ \Pr_{pl}( G \mid \xb ) : G\in \mathcal{G}_p\}$ using a method that involves two phases.

We consider three types of graphs (see Figure \ref{fig:plot_graphs}):
\begin{itemize}
\item[1.] \textbf{Random:}  A graph in which edges were randomly generated from the prior \eqref{eq:graphprior} with $\beta = 0.4$.
\item[2.] \textbf{Cluster:} A graph with two clusters (connected components) each with $p=5$ vertices. The edges in both clusters were randomly generated from the prior \eqref{eq:graphprior} with $\beta = 0.6$.
 \item[3.] \textbf{Scale-free:} A graph sampled from a power-law degree distribution with the Barab\'asi-Albert algorithm \cite{albert2002statistical}.
 \end{itemize}
We also consider graphs with $p=20$ vertices that have two connected components with $10$ vertices and the same edge structure of type ``Random'', ``Cluster'', or ``Scale-free''. We simulated binary contingency tables with $p\in \{10,20\}$ variables that comprise $n \in \{ 200, 500, 1000 \}$ samples from random graphs of these three types. We repeated the simulation experiment that involves the generation of $18$ contingency tables $50$ times. We performed all computations with the {\tt R} package {\tt BDgraph} \cite{BDgraph, mohammadi2015bdgraph}. For each contingency table we generated, we run the BDMCMC and the HC algorithms using a uniform prior on the graph space $\mathcal{G}_p$ (the equivalent of setting $\beta = 0.5$ in \eqref{eq:graphprior}) starting from the empty graph. The BDMCMC algorithm was run for $100,000$ iterations. The first $60,000$ iterations were discarded as burn-in.

We estimated the structure of the true graph based on model averaging \cite{madigan1996bayesian} of the graphs sampled by the BDMCMC algorithm. We calculate the posterior inclusion probabilities of edges \eqref{posterior-edge}, and determine the median graph whose edges have posterior inclusion probabilities greater than $0.5$. The structure of the true graph was estimated with the HC algorithm based on the ``and" and the ``or" criteria in the first phase of the algorithm \cite{pensar2016marginal}. 

\begin{figure} [!ht] 
\centering
\subfigure[Random]{
   \centering
   \includegraphics[width=0.3\textwidth]{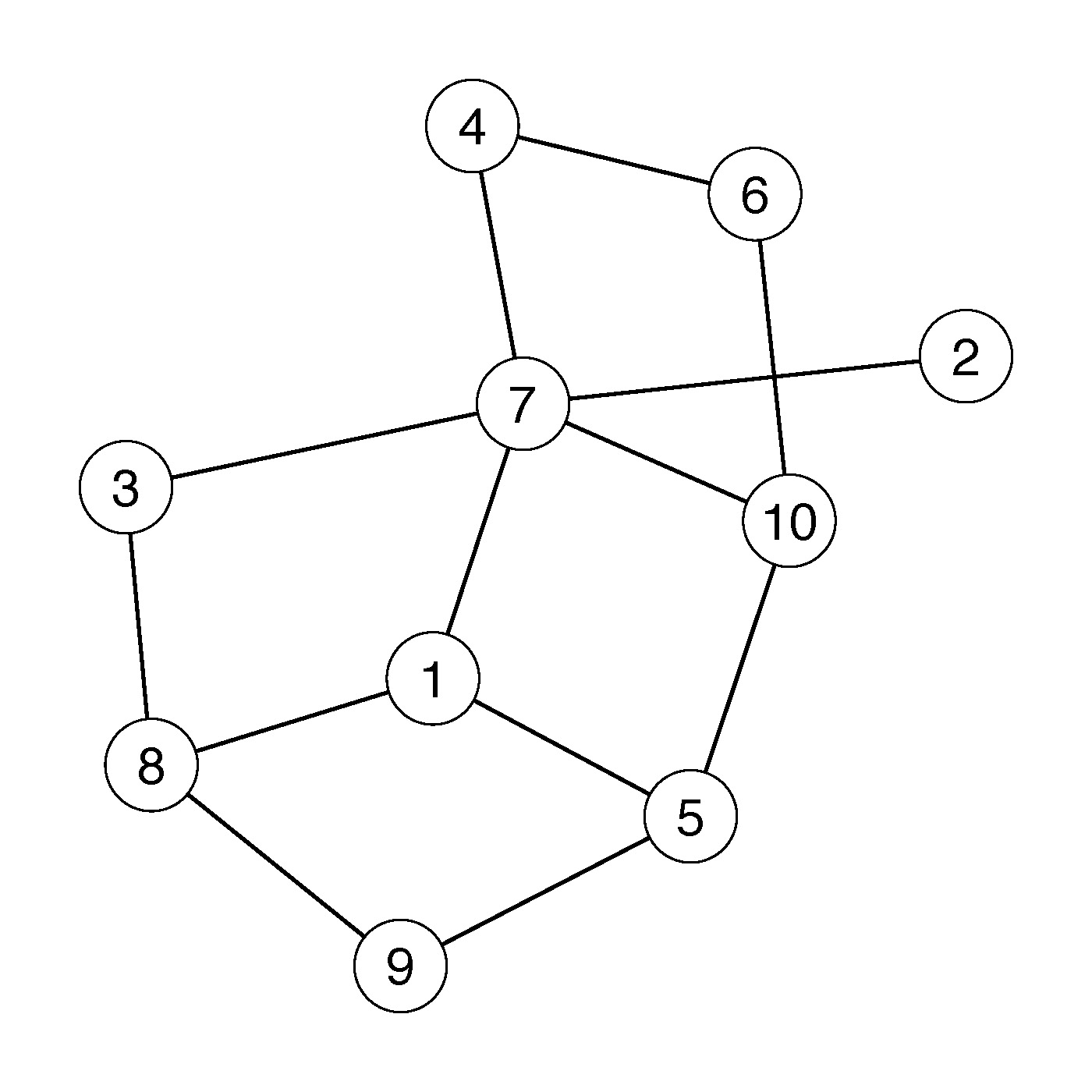}
   \label{fig:Random}
}
\centering
\subfigure[Cluster]{
   \centering
   \includegraphics[width=0.3\textwidth]{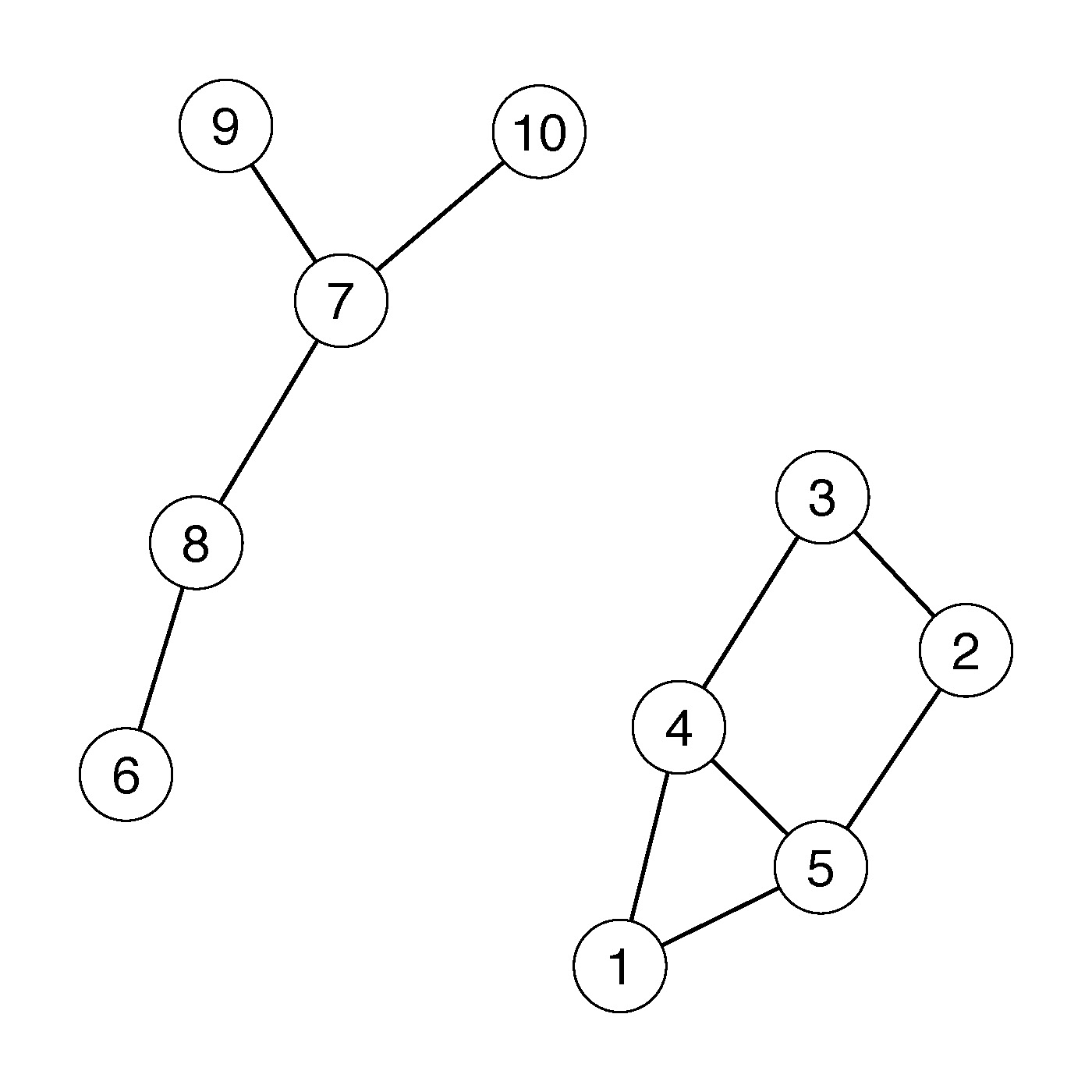}
   \label{fig:Cluster}
}
\centering
\subfigure[Scale-free]{
   \centering
   \includegraphics[width=0.3\textwidth]{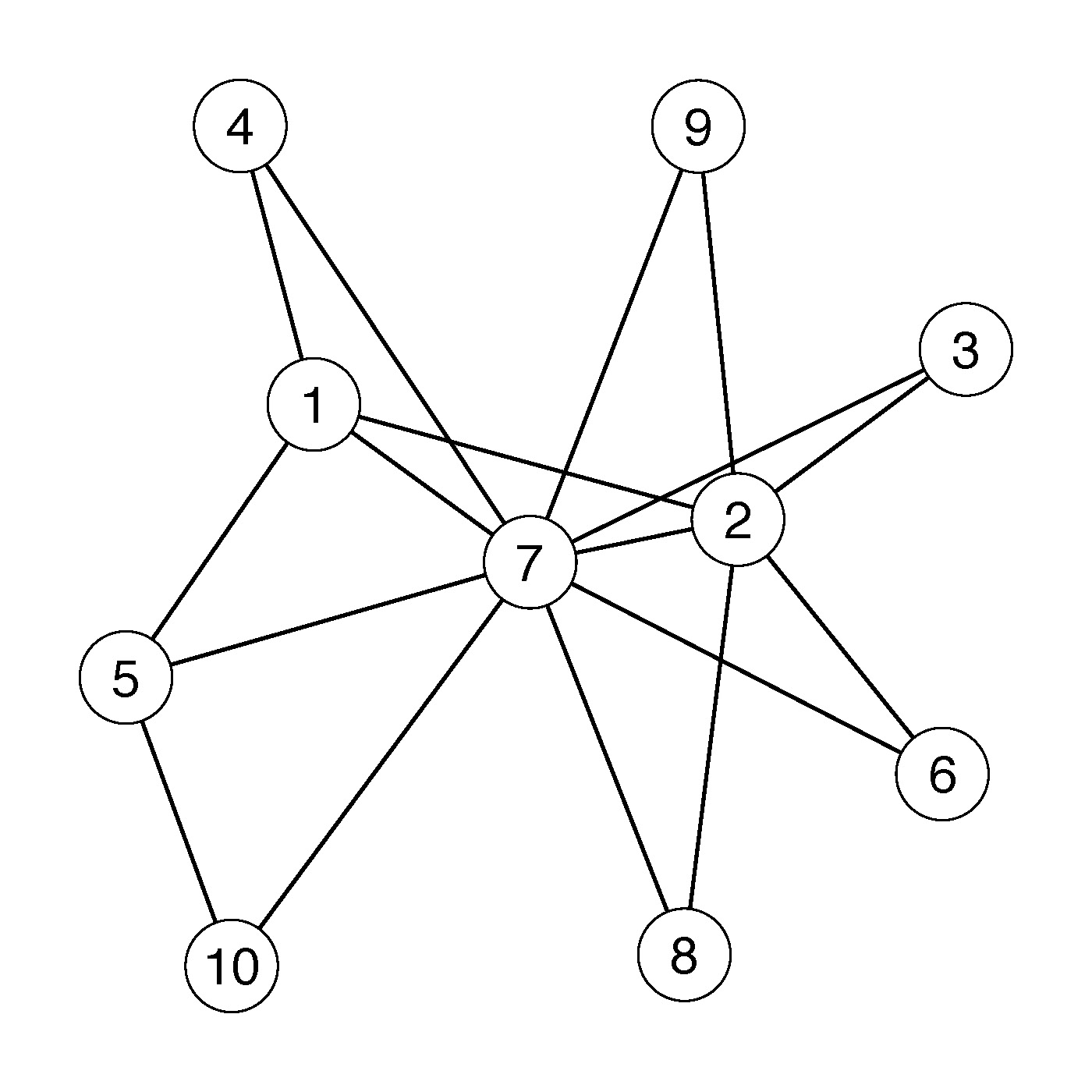}
   \label{fig:Scale-free}
}
\caption[]{Example graphs with $p=10$ vertices used in the simulation study from Section \ref{sec:simulation}.}
\label{fig:plot_graphs}
\end{figure}

We evaluate the performance of the two algorithms in recovering the structure of the true graphs using the $F_1$-score measure \cite{baldi2000assessing},  
\begin{eqnarray}
\label{f1}
F_1\mbox{-score} = \frac{2 \mbox{TP}}{2 \mbox{TP + FP + FN}},
\end{eqnarray}
\noindent and the Structure Hamming distance (SHD) \cite{tsamardinos2006max},
\begin{eqnarray}
\label{SHD}`
\mbox{SHD} = \mbox{FP + FN},
\end{eqnarray}
\noindent where TP, TN, FP and FN are the number of true positives, true negatives, false positives and false negatives, respectively. The values of the $F_1$-score range between $0$ and $1$, and the values of the SHD are positive. A better performance in recovering the true graph is associated with larger values of the $F_1$-score, and with smaller values of the SHD.

\begin{table*}[!ht] \footnotesize 
\vspace{0em}
\caption{ \label{table:measures}
Results from the simulation study from Section \ref{sec:simulation}. Binary tables with  $p\in \{10,20\}$ variables and  $n\in \{200,500,1000\}$ samples were generated from graphical models defined by three types of graphs: ``Random'', ``Cluster'', and ``Scale-free''. This table reports means and standard deviations of the $F_1$-score \eqref{f1} and SHD \eqref{SHD} measures for the accuracy in recovering the structure of the true graph across $50$ replicate binary tables generated for every combination of $p$, $n$ and graph type. The best performing models in terms of the $F_1$-score and SHD are shown in boldface.}
\centering
\begin{tabular}{l*{9}{l}l}  
\toprule
                    &      & \multicolumn{3}{c}{$F_1$-score}    & \multicolumn{4}{c}{SHD}      \\
                                                          \cmidrule{3-5}                       \cmidrule{7-9}               
 p                  & n    & BDMCMC   & HC(or)    & HC(and)      && BDMCMC & HC(or)  & HC(and)        \\
\hline   
                                                                                                                                           \\
\multicolumn{9}{c}{ Random }                                                                                                               \\
\\
\multirow{3}{*}{10} & 200   & \textbf{0.7}(0.12)  & 0.68(0.11)         & 0.57(0.11)  && \textbf{8.2}(2.9) & 8.7(2.9)         & 10.6(3)       \\  
                    & 500   & \textbf{0.8}(0.1)   & 0.78(0.1)          & 0.7(0.1)    && \textbf{5.8}(2.5) & 6.2(2.6)         & 8(2.6)        \\  
                    & 1000  & \textbf{0.87}(0.08) & 0.86(0.08)         & 0.8(0.09)   && \textbf{3.9}(2.2) & 4(2.1)           & 5.6(2.2)      \\  
\\
\multirow{3}{*}{20} & 200   & \textbf{0.7}(0.08)  & 0.69(0.07)         & 0.58(0.09)  && 17.5(4.6)        & \textbf{17.3}(3.9)& 21.5(5.7)     \\  
                    & 500   & \textbf{0.8}(0.08)  & 0.79(0.08)         & 0.7(0.09)   && \textbf{11.9}(4.2)& 12.5(4.5)        & 16.6(4.4)     \\  
                    & 1000  & \textbf{0.85}(0.07) & \textbf{0.85}(0.07) & 0.78(0.08)  && \textbf{8.9}(4)   & 9.3(3.8)         & 12.6(4)       \\  
                                                                                                                                           \\
\multicolumn{9}{c}{ Cluster }                                                                                                              \\
                                                                                                                                           \\
\multirow{3}{*}{10} & 200   & \textbf{0.76}(0.13) & 0.75(0.13)         & 0.66(0.14)  && \textbf{4.5}(2)   & 4.6(1.8)         & 5.8(2)        \\  
                    & 500   & \textbf{0.86}(0.1)  & 0.83(0.13)         & 0.75(0.15)  && \textbf{2.7}(1.7) & 3.3(2.1)         & 4.6(2.4)      \\  
                    & 1000  & \textbf{0.91}(0.09) & 0.9(0.09)          & 0.87(0.09)  && \textbf{1.7}(1)   & 1.8(1.1)         & 2.4(1.3)      \\  
\\
\multirow{3}{*}{20} & 200   & 0.69(0.13)         & \textbf{0.71}(0.11) & 0.63(0.14)  && 14.8(10.2)       & \textbf{11.7}(4.2)& 13.5(4.9)     \\  
                    & 500   & \textbf{0.86}(0.08) & 0.84(0.08)         & 0.76(0.11)  && \textbf{5.9}(2.8) & 6.7(3.2)         & 9.5(4.1)      \\  
                    & 1000  & \textbf{0.93}(0.06) & 0.92(0.06)         & 0.87(0.08)  && \textbf{3.3}(2.7) & 3.6(2.7)         & 5.9(2.9)      \\  
                                                                                                                                           \\
\multicolumn{9}{c}{ Scale-free }                                                                                                           \\
                                                                                                                                           \\
\multirow{3}{*}{10} & 200   & \textbf{0.67}(0.12) & 0.66(0.1)          & 0.56(0.1)   && \textbf{8.5}(2.6) & \textbf{8.5}(2.1) & 10(1.9)       \\  
                    & 500   & \textbf{0.73}(0.11) & \textbf{0.73}(0.11) & 0.62(0.12)  && \textbf{6.9}(2.4) & \textbf{6.9}(2.2) & 8.8(2.2)      \\  
                    & 1000  & 0.8(0.07)          & \textbf{0.81}(0.08) & 0.7(0.1)    && 5.3(1.7)         & \textbf{5.2}(1.9) & 7.4(1.9)      \\  
\\
\multirow{3}{*}{20} & 200   & \textbf{0.63}(0.13) & \textbf{0.63}(0.13) & 0.53(0.09)  && 21.3(12.7)       & \textbf{19.1}(5.6)& 21.4(3.7)     \\  
                    & 500   & \textbf{0.74}(0.08) & \textbf{0.74}(0.09) & 0.63(0.08)  && 14(3.6)          & \textbf{13.8}(3.8)& 17.7(3.2)     \\  
                    & 1000  & \textbf{0.78}(0.09) & \textbf{0.78}(0.09) & 0.7(0.09)   && 11.8(4.3)        & \textbf{11.4}(4.3)& 14.7(3.9)     \\  
\bottomrule
\end{tabular}
\end{table*}

The results are summarized in Table \ref{table:measures}. For most simulation experiments, the BDMCMC algorithm has an advantage over the HC algorithm especially for the $F_1$-score. ROC curves showing the performance of the BDMCMC algorithm are presented in Figure \ref{fig:Rocplot}.
\begin{figure} [ht]
\centering
    \includegraphics[width=12cm,height=5cm]{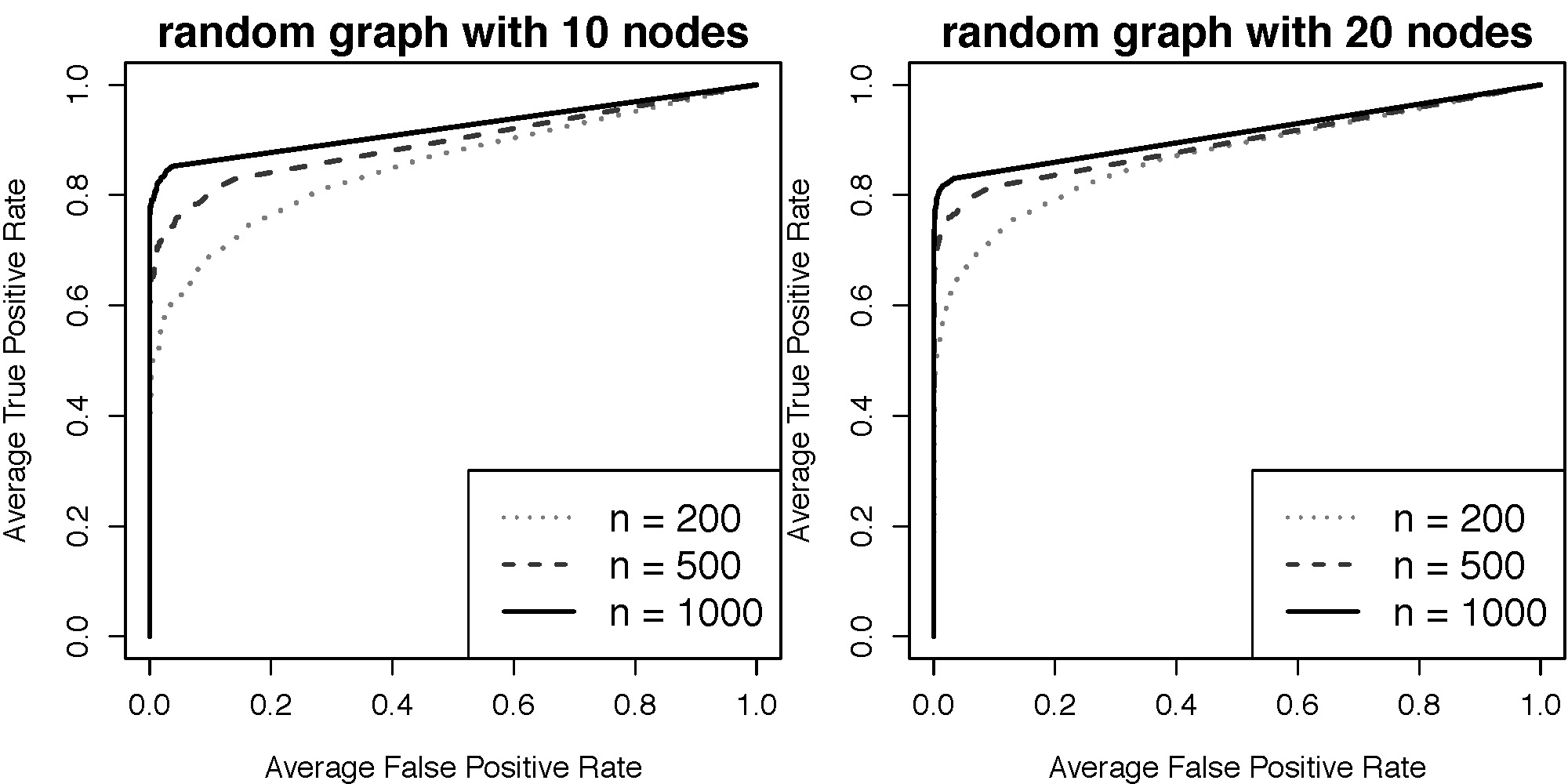} 
    \includegraphics[width=12cm,height=5cm]{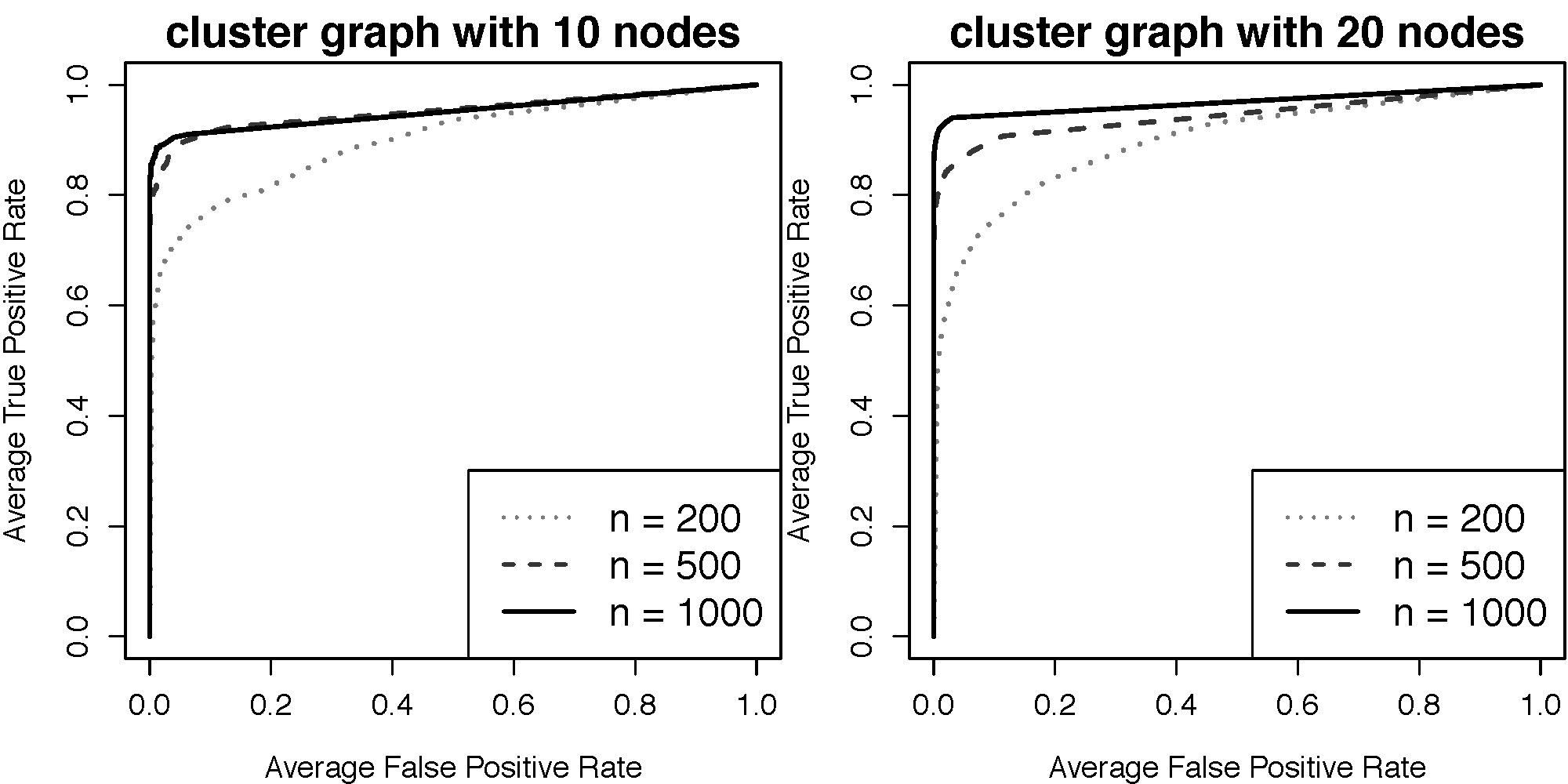}
    \includegraphics[width=12cm,height=5cm]{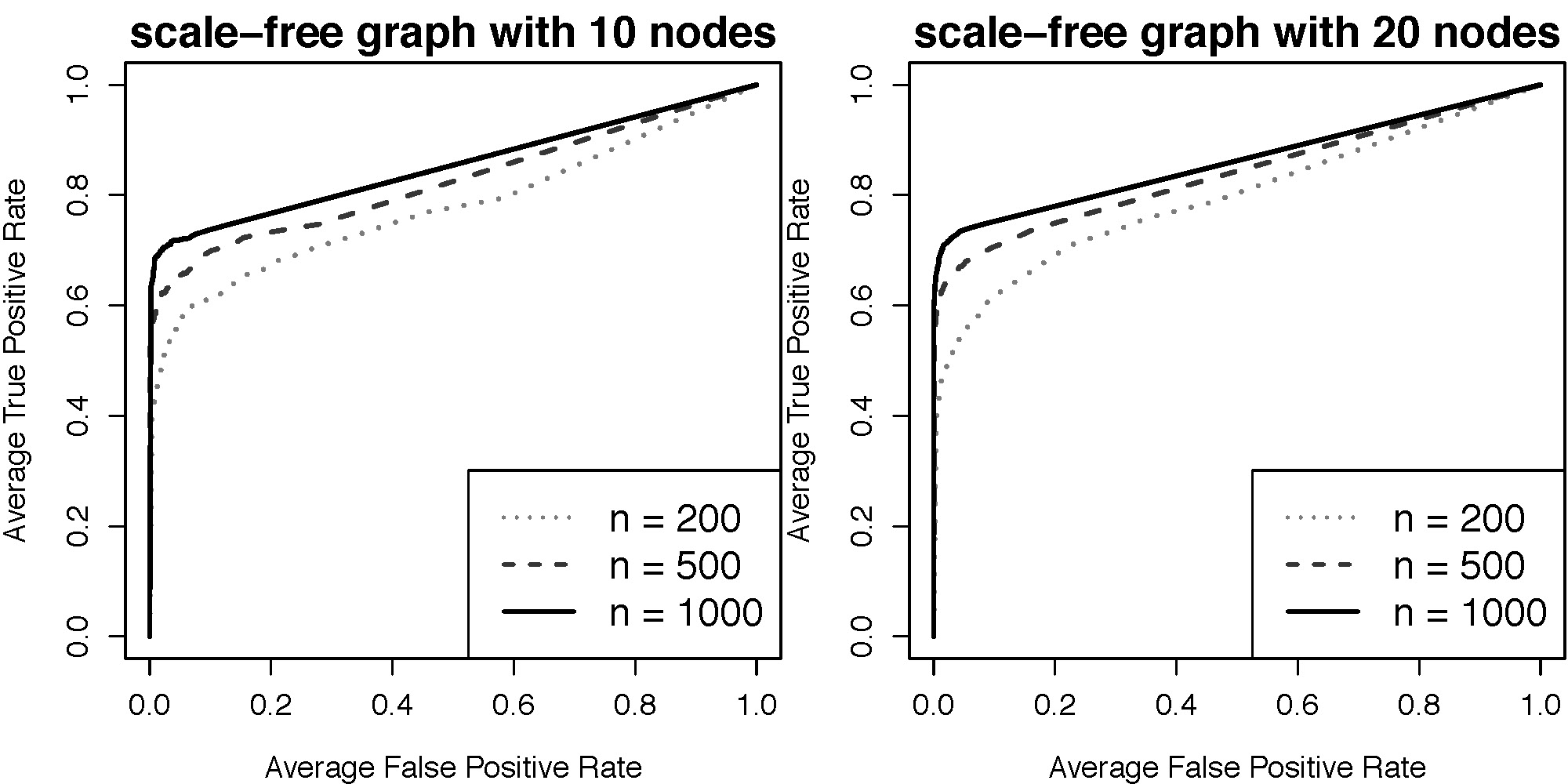}
\caption{ROC curves showing the performance of the BDMCMC algorithm in recovering the structure of the true graph in the simulation study from Section \ref{sec:simulation}. Binary tables with $p\in \{10,20\}$ variables and $n\in \{200,500,1000\}$ samples were simulated from graphs of three types (``Random'', ``Cluster'', and ``Scale-free'').}
\label{fig:Rocplot}
\end{figure}

\section{Analysis of the Geolocated Twitter Data}
\label{sec:twitteranalysis}

We come back to the $p=214$ dimensional binary contingency table constructed from geotweets that was described in Section \ref{sec:twitterdata}. We use the BDMCMC algorithm to sample graphs from the marginal posterior distribution \eqref{eq:postmpl} on $\mathcal{G}_{214}$. We employ the prior \eqref{eq:graphprior} with $\beta = 1/{214 \choose 2} = 4.388\times 10^{-5}$.  Under this prior, the expected number of edges is $1$, thus sparser graphs receive larger prior probabilities compared to denser graphs. We performed all computations on a cluster with $7$ compute nodes, each with $48$ Intel Xeon $2.6$ GHz cores with a Linux operating system.

\begin{figure} [ht] 
\centering
\includegraphics[width=1\textwidth]{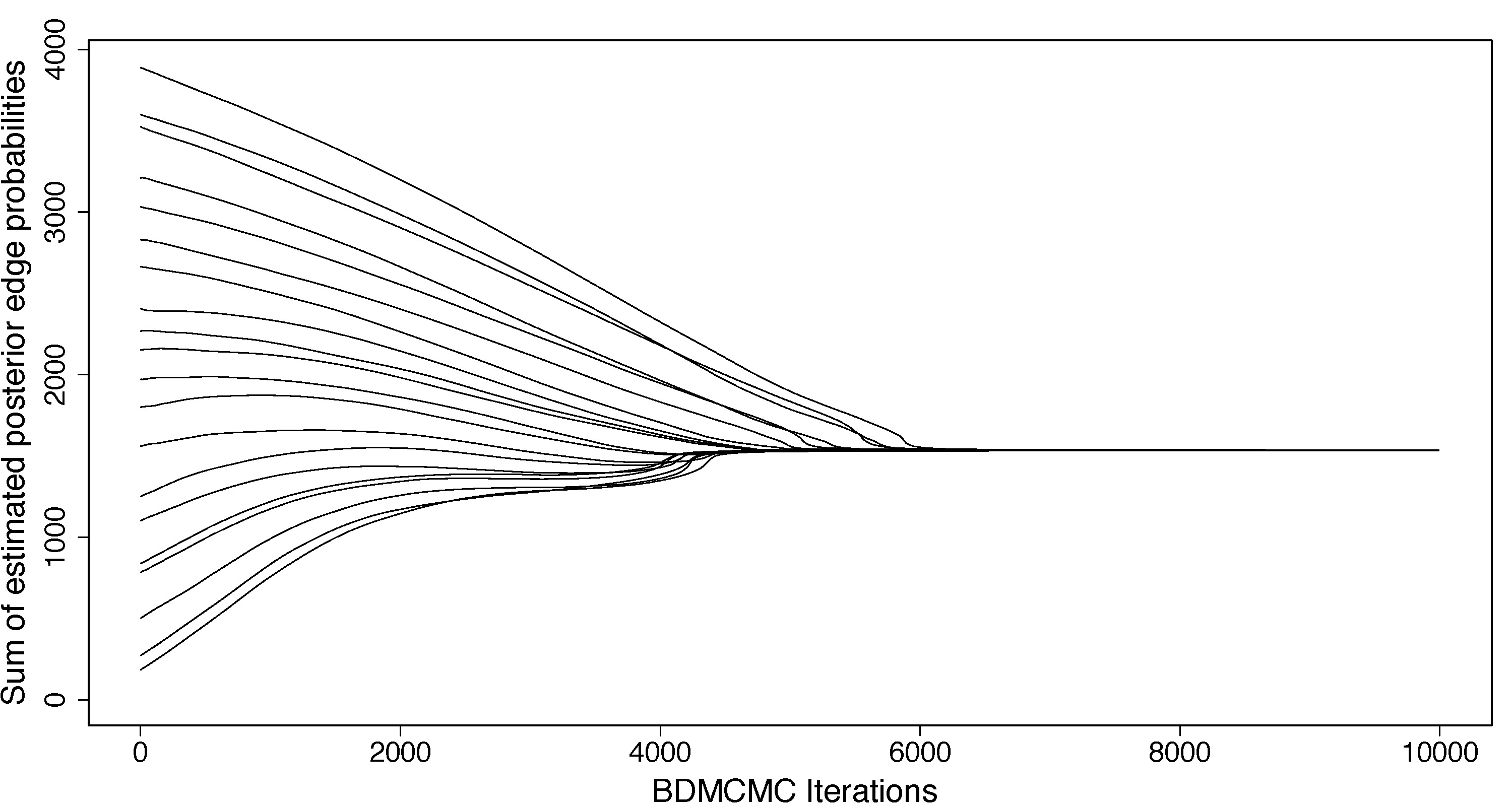}
\caption{ \label{fig:convergence_plinks} 
Convergence plot of the BDMCMC algorithm for the Twitter data introduced in Section \ref{sec:twitterdata} and analyzed in Section \ref{sec:twitteranalysis}. The plot shows the sum of the estimated posterior edge inclusion probabilities ($y$ axis) against iteration number ($x$ axis) from $20$ runs of the BDMCMC algorithm starting from $20$ different graphs sampled from graph space $\mathcal{G}_{214}$ with different number of edges.}
\end{figure}

\begin{figure} [ht] 
\centering
\includegraphics[width=1\textwidth]{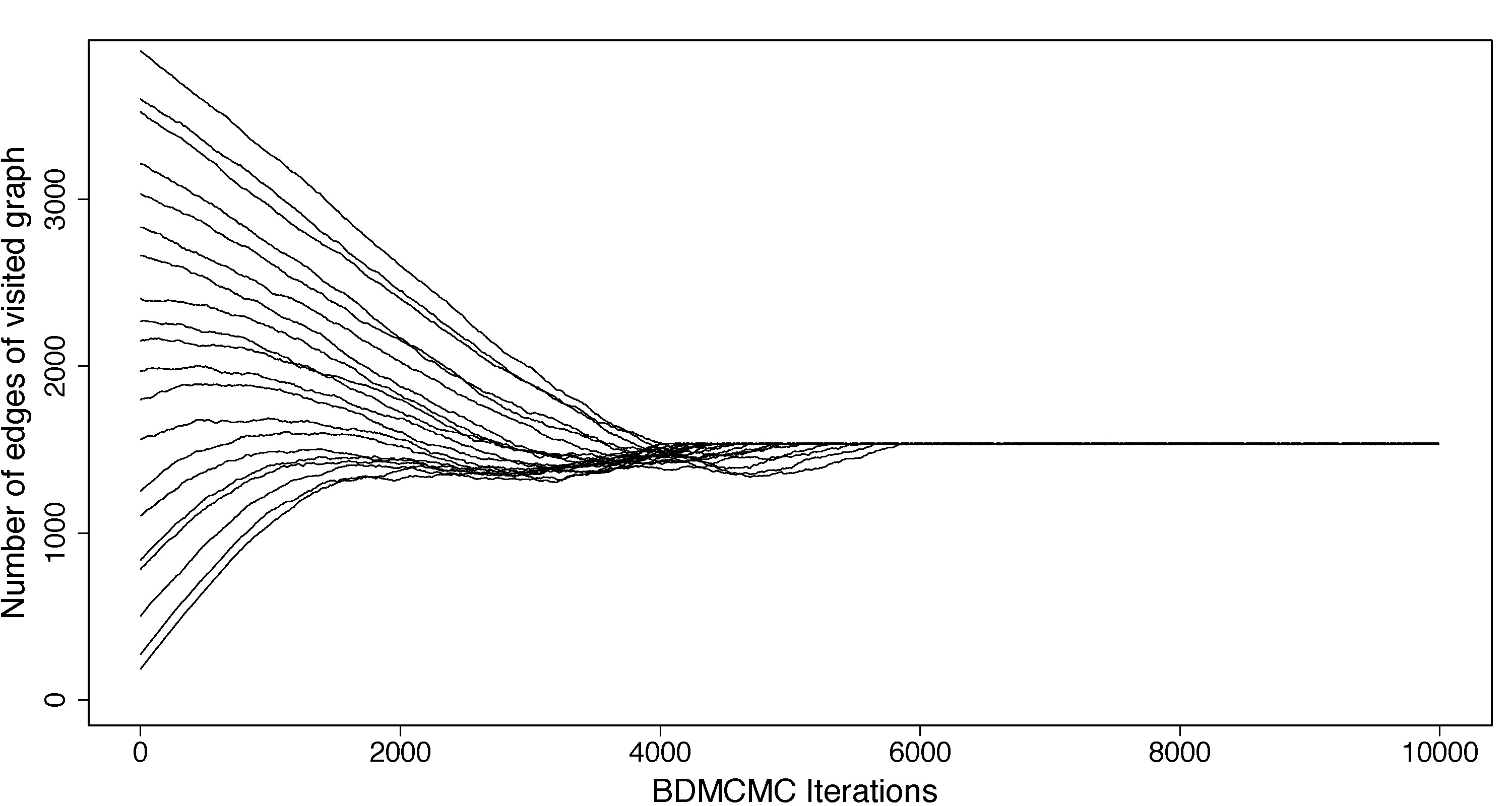}
\caption{ \label{fig:convergence_size_g} 
Convergence plot of the BDMCMC algorithm for the Twitter data introduced in Section \ref{sec:twitterdata} and analyzed in Section \ref{sec:twitteranalysis}. The plot shows the number of edges included in the sampled graphs ($y$ axis) against iteration number ($x$ axis) from $20$ runs of the BDMCMC algorithm starting from $20$ different graphs sampled from graph space $\mathcal{G}_{214}$ with different number of edges.}
\end{figure}

First, we want to gain some understanding of the ability of the BDMCMC algorithm to move towards graphs with large posterior probabilities in $\mathcal{G}_{214}$. To this end, we sample $20$ graphs $\{G_i\}_{i=1}^{20}$ having increasing number of edges: $G_i$ has a number of edges randomly sampled from $(200(i-1),200i)$. The resulting set of graphs ranges from most sparse ($G_1$) to most dense ($G_{20}$). Starting from each graph $G_i$, $i=1,\ldots,20$, we ran the BDMCMC algorithm for $10,000$ iterations. Figures \ref{fig:convergence_plinks} and \ref{fig:convergence_size_g} show the sum of the estimated posterior edge inclusion probabilities and the number of edges included in the sampled graphs against iteration number. After $7,000$ iterations in each of the $20$ runs, the BDMCMC algorithm seems to have reached the same neighborhood of graphs. Thus, although the number of graphs in $\mathcal{G}_{214}$ is extremely large ($\approx 10^{6861}$), the BDMCMC algorithm seems to be very efficient in identifying graphs with high posterior probability. 

Next, we ran the BDMCMC algorithm for $400,000$ iterations using parallel calculations of the birth and death rates from a starting graph sampled from the prior \eqref{eq:graphprior} on $\mathcal{G}_{214}$. The first $200,000$ iterations were discarded as burn-in. Figures \ref{fig:longrun_convergence_plinks} and \ref{fig:longrun_convergence_size_g} show the BDMCMC algorithm seems to have reached convergence in less than $10,000$ iterations.

\begin{figure} [ht] 
\centering
\includegraphics[width=1\textwidth]{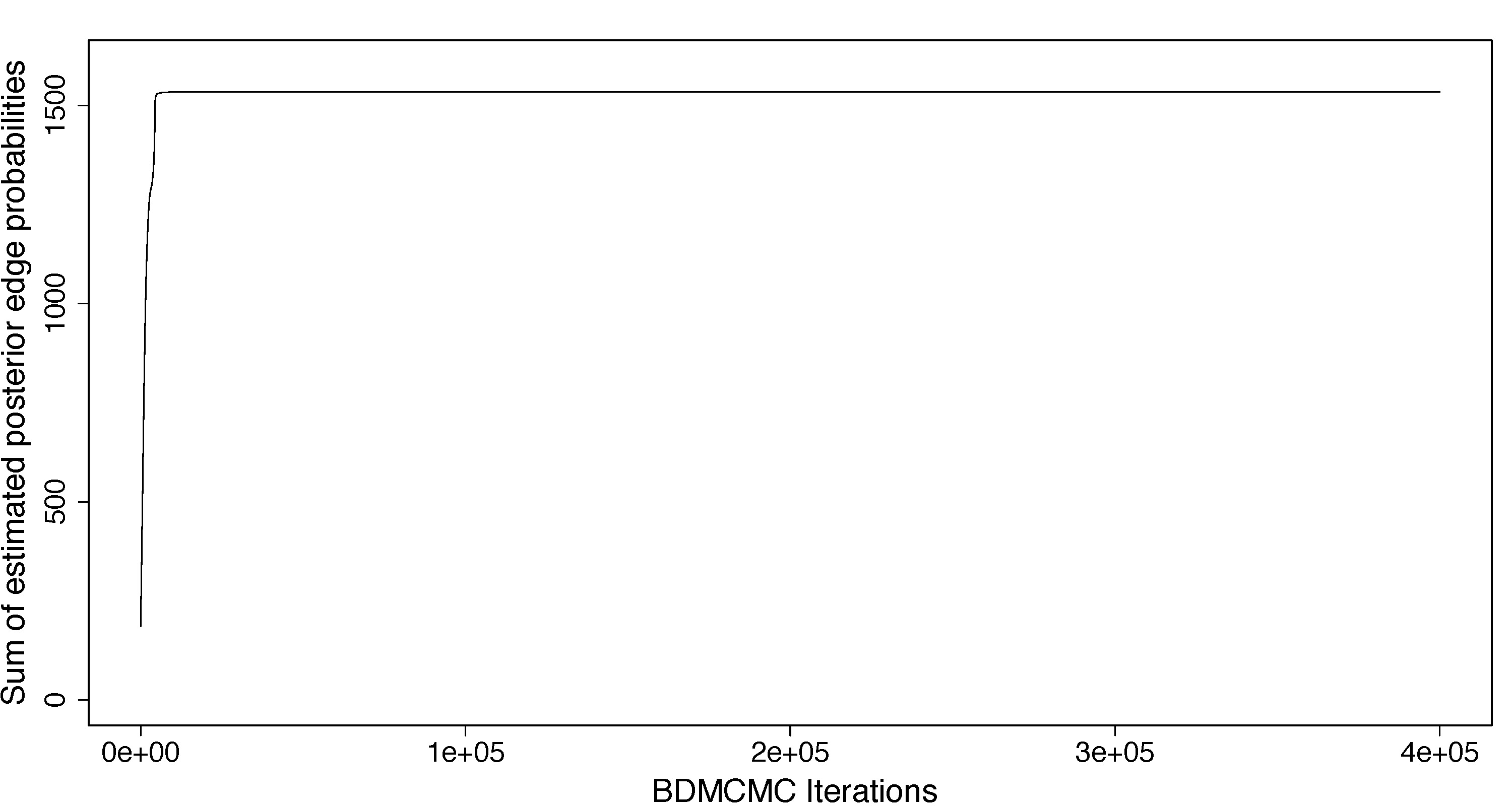}
\caption{ \label{fig:longrun_convergence_plinks} 
Convergence plot for a longer run of 400,000 iterations of the BDMCMC algorithm for the Twitter data introduced in Section \ref{sec:twitterdata} and analyzed in Section \ref{sec:twitteranalysis}. The plot shows the sum of the estimated posterior edge inclusion probabilities ($y$ axis) against iteration number ($x$ axis).}
\end{figure}

\begin{figure} [ht] 
\centering
\includegraphics[width=1\textwidth]{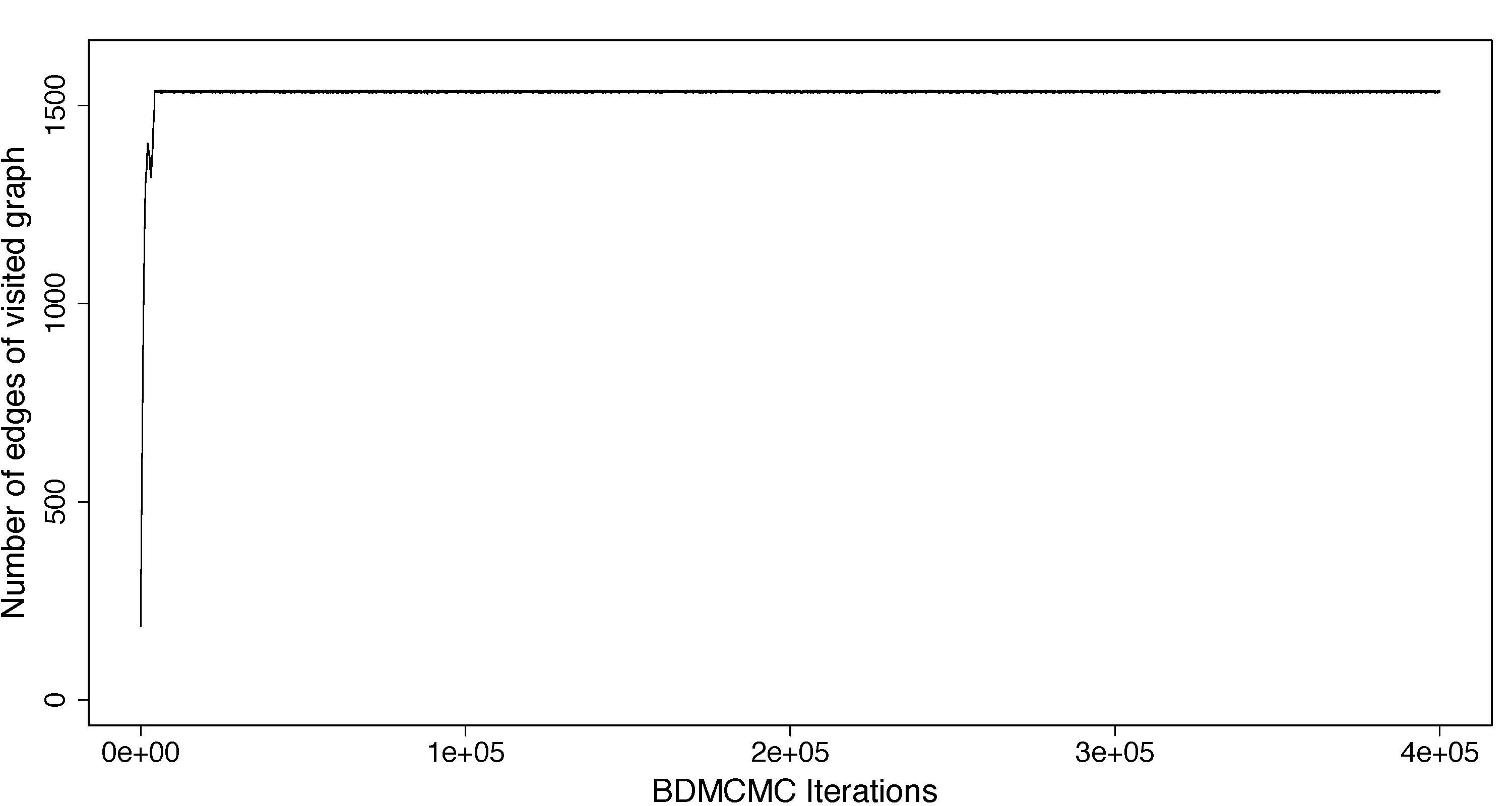}
\caption{ \label{fig:longrun_convergence_size_g} 
Convergence plot for a longer run of $400,000$ iterations of the BDMCMC algorithm for the Twitter data introduced in Section \ref{sec:twitterdata} and analyzed in Section \ref{sec:twitteranalysis}. The plot shows the number of edges included in the sampled graphs ($y$ axis) against iteration number ($x$ axis).}
\end{figure}

\begin{figure} [ht] 
  \begin{center}
    \includegraphics[width=13cm]{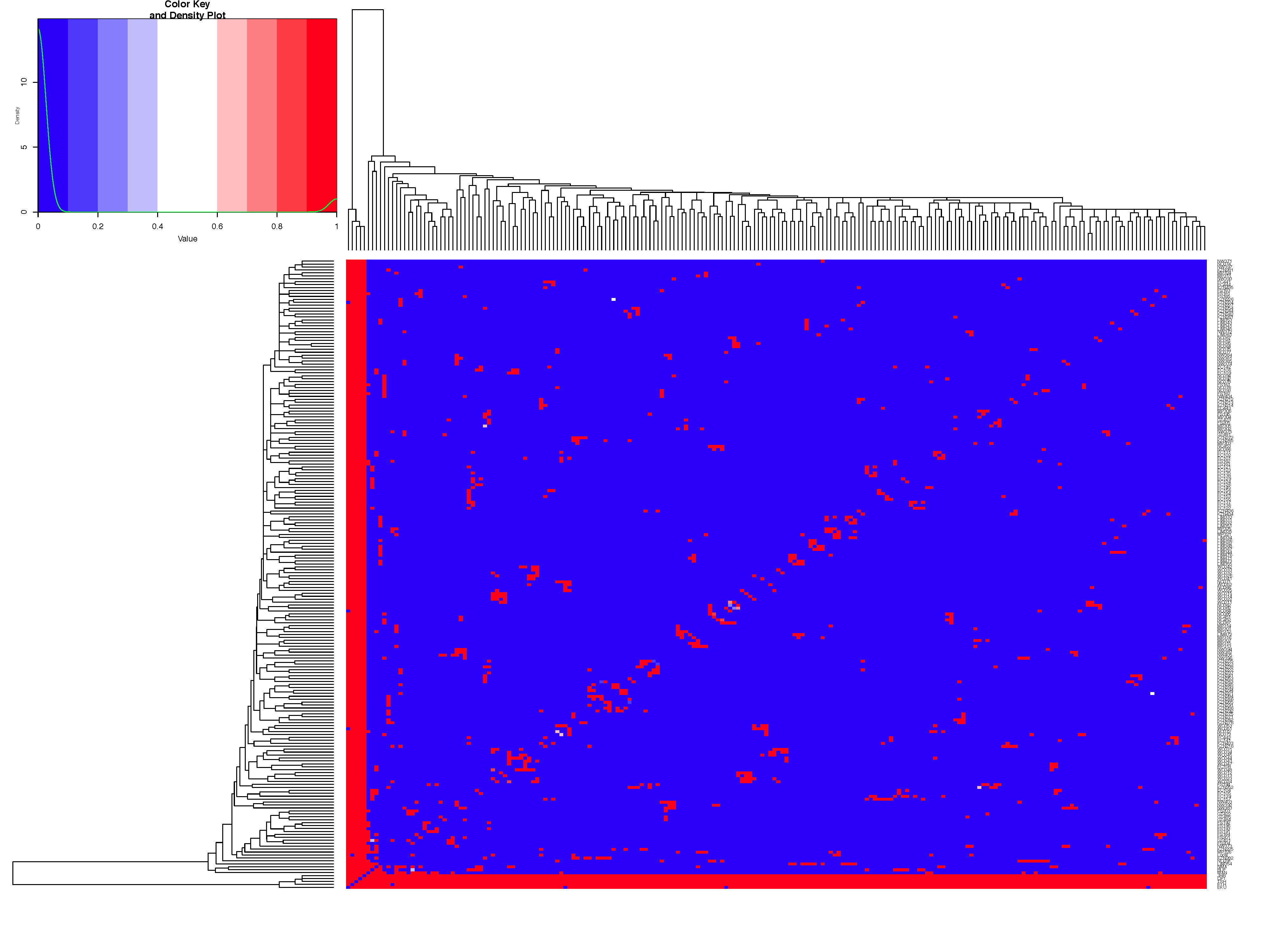}
  \end{center}
  \caption{\label{fig:heatmap} Heatmap of the $214\times 214$ matrix of posterior inclusion probability of edges for the Twitter data introduced in Section \ref{sec:twitterdata} and analyzed in Section \ref{sec:twitteranalysis}. The rows are shown in increasing order with respect to the row sums from top to bottom. The columns are shown in the same order from left to right. Red marks edges with posterior inclusion probability very close to $1$, while blue marks edges with posterior inclusion probability close to zero. The names of the variables involved are shown in the right column.}
\end{figure}

We estimate the posterior inclusion probabilities \eqref{posterior-edge} of the ${214 \choose 2} = $ $22,791$ edges. We show a heatmap of the matrix of the estimated posterior edge inclusion probabilities in Figure \ref{fig:heatmap}. Most of the estimated posterior edge inclusion probabilities are zero: $21,138$ ($92.78$\%). A number of 12, 5 and 7 edges have estimated posterior inclusion probabilities in $(0,0.5)$, $[0.5,0.9)$ and $[0.9,0.1)$, respectively. The remaining $1522$ ($6.65$\%) have estimated posterior inclusion probabilities equal to $1$. We use the median graph which includes the 1534 edges with estimated posterior inclusion probabilities greater than $0.5$ as our estimate of the conditional independence graph. Henceforth we refer to this graph as the South Africa (SA) Twitter graph. 

\begin{figure} [!ht] 
\begin{center}
\includegraphics[width=12cm]{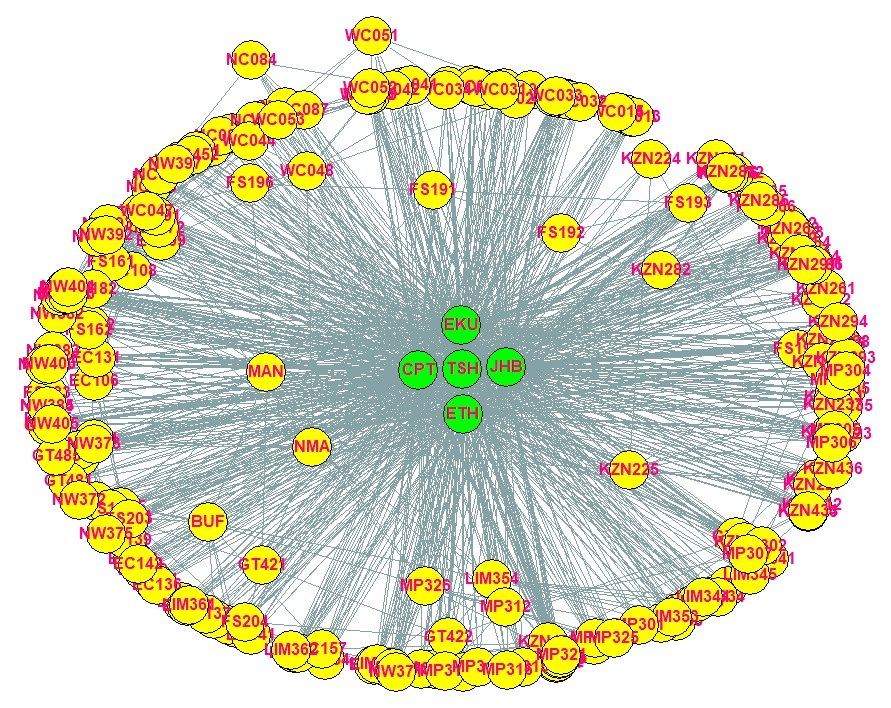}
\end{center}
\caption{ \label{fig:hubnetwork} 
Rendering of the SA Twitter graph. The five hubs are shown in green, and the remaining vertices are shown in yellow.}
\end{figure}

A rendering of the SA Twitter graph is presented in Figure \ref{fig:hubnetwork}. The 213 district municipalities of South Africa are denoted by their identifiers \--- see Table \ref{tab:summarygeodem}. This table provides the identifier, complete name, province to which it belongs, area, population size, and density for each municipality. The $214$-th vertex of this graph is associated with the Local (yes/no) variable. We explore the SA Twitter graph using four centrality measures \cite{imai-2017} that capture the extent to which a vertex is connected to other vertices, and occupies a central position in the structure of the network: (i) {\it degree} counts the number of edges that originate from a given vertex; (ii) {\it closeness} measures how close is a vertex from each one of the other vertices; (iii) {\it betweenness} finds vertices that connect other vertices (i.e., belong to the shortest paths connecting pairs of vertices); and (iv) {\it page rank} defines more central vertices based on a voting process which allocates votes to a vertex based on other connected vertices, and it is determined through an iterative algorithm. Barplots of the largest $10$ values of each of the four centrality measures are presented in Figures \ref{fig:degreetop10}, \ref{fig:closenesstop10}, \ref{fig:betweennesstop10}, and \ref{fig:pagerank10}.

\begin{figure} [ht] 
\begin{center}
\includegraphics[width=7cm]{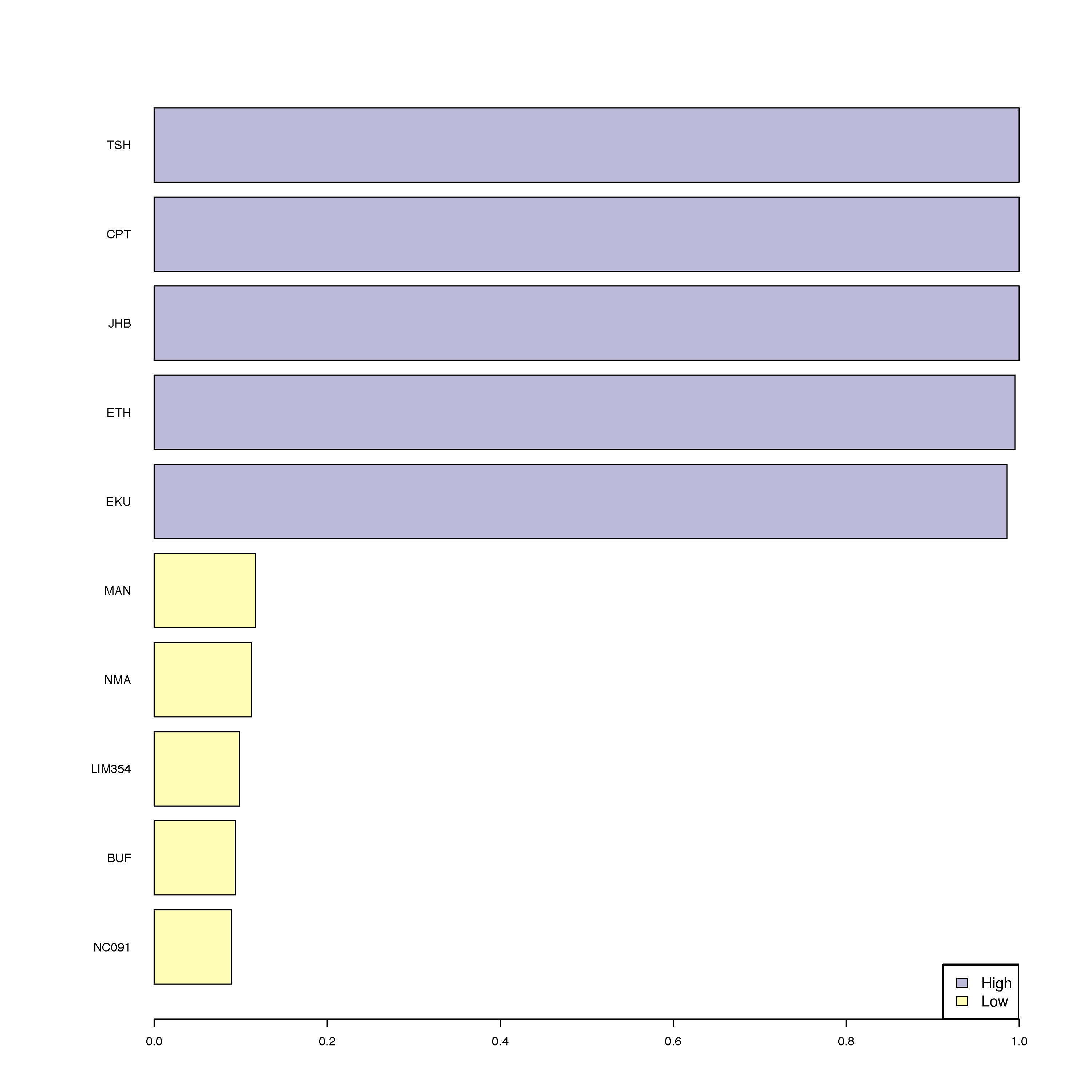}
\end{center}
\caption{ \label{fig:degreetop10} 
Barplot showing the top $10$ municipalities with the largest degree in the SA Twitter graph. The five hubs are shown in purple.}
\end{figure}

\begin{figure} [ht] 
\begin{center}
\includegraphics[width=7cm]{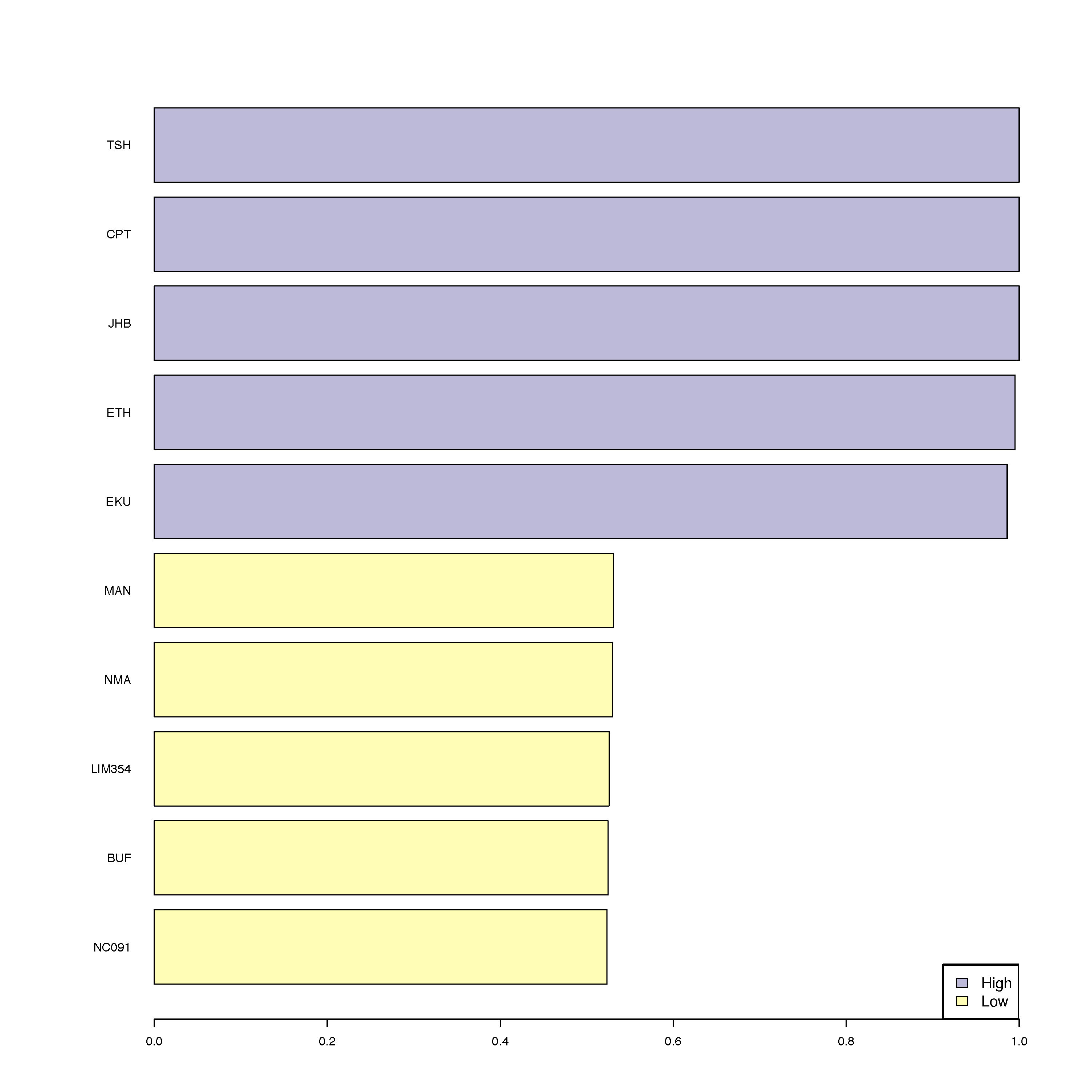}
\end{center}
\caption{ \label{fig:closenesstop10} 
Barplot showing the top $10$ municipalities with the largest closeness in the SA Twitter graph. The five hubs are shown in purple.}
\end{figure}

\begin{figure} [ht] 
\begin{center}
\includegraphics[width=7cm]{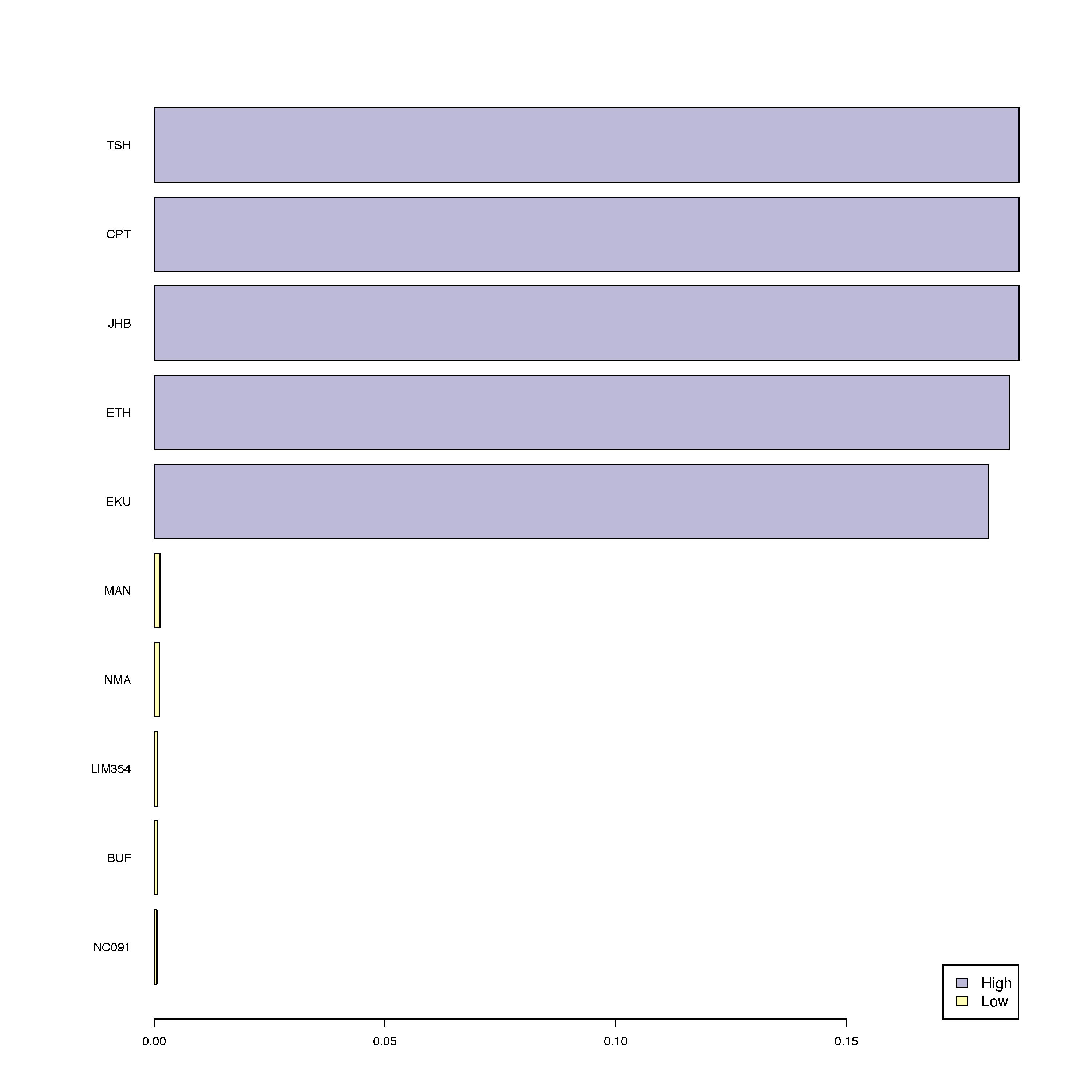}
\end{center}
\caption{ \label{fig:betweennesstop10} 
Barplot showing the top $10$ municipalities with the largest betweenness in the SA Twitter graph. The five hubs are shown in purple.}
\end{figure}

\begin{figure} [ht] 
\begin{center}
\includegraphics[width=7cm]{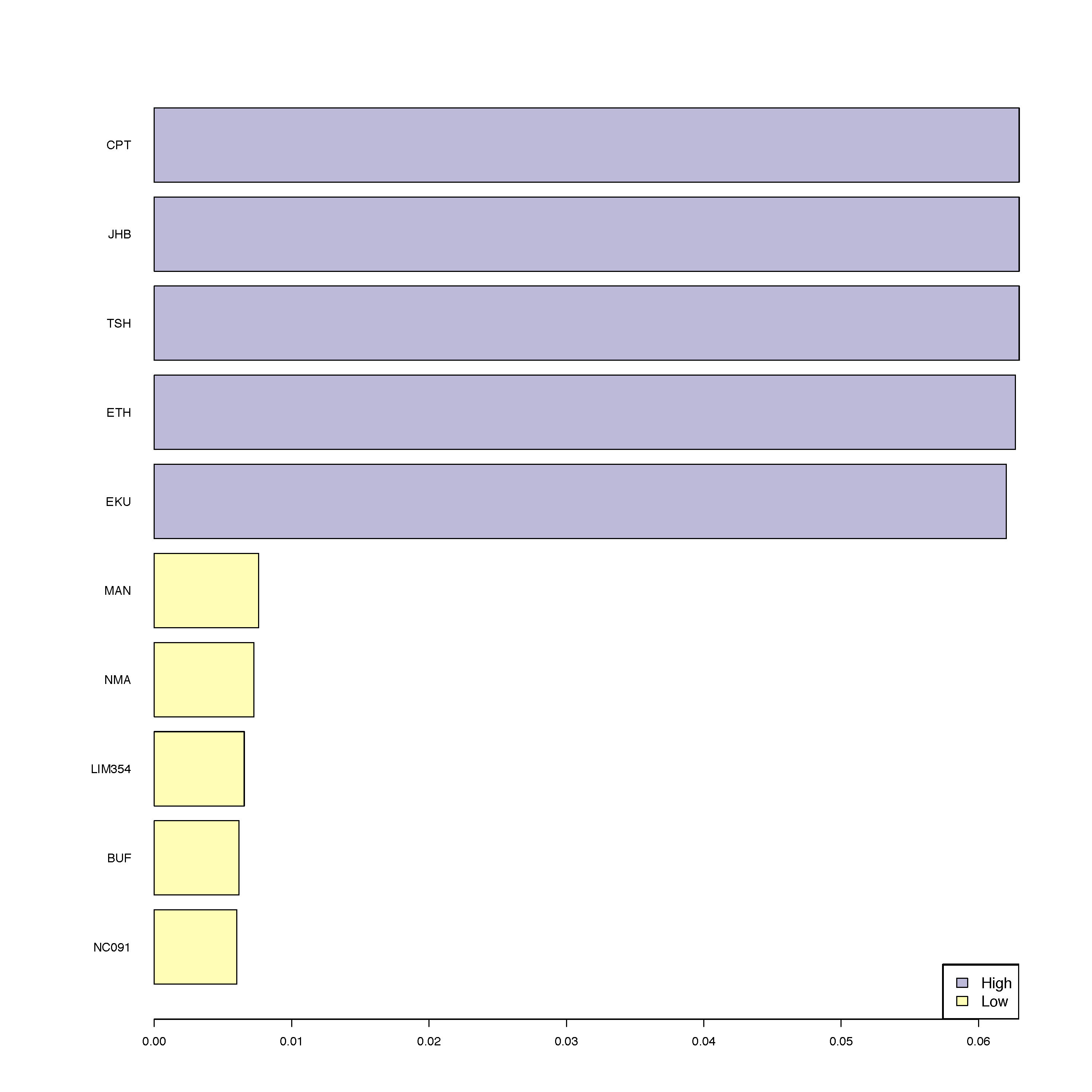}
\end{center}
\caption{ \label{fig:pagerank10} 
Barplot showing the top $10$ municipalities with the largest page rank in the SA Twitter graph. The five hubs are shown in purple.}
\end{figure}

For each of the four centrality measures, their top five largest values correspond with the following municipalities: Johannesburg (JHB, Gauteng), Ekurhuleni (EKU, Gauteng), Tshwane (TSH, Gauteng), eThekwini (ETH, KwaZulu-Natal), and Cape Town (CPT, Western Cape). The next five largest values also correspond with the same five municipalities for all four measures: Mangaung (MAN, Free State), Nelson Mandela Bay (NMA, Eastern Cape), Polokwane (LIM354, Limpopo), Buffalo City (BUF, Eastern Cape), and Sol Plaatjie (NC091, Northern Cape). We remark that the values of centrality measures for JHB, EKU, TSH, ETH and CPT are significantly larger than the values for MAN, NMA, LIM354, BUF and NC091. For example, the degrees for the first group are $213$, $210$, $213$, $212$, and $213$, while the degrees for the second group are $25$, $24$, $21$, $20$, and $19$.

As such, MAN, NMA, LIM354, BUF and NC091 are the five key hubs of the SA Twitter graph. Their geographical location is mapped in Figure \ref{fig:hubmap}, while Table \ref{tab:hub} gives summary information about them. Three hubs (JHB, TSH, EKU) are located in the Johannesburg/Pretoria area which represents the region of South Africa in which more than 11 million people reside ($2015$ South African National Census) either permanently, or temporarily to find employment in factories or gold mines. The other two hubs are located around the cities of Cape Town and Durban which, together with Johannesburg and Pretoria, are among the largest South African cities. A great number of local and international travelers visit these five hubs for shorter or longer periods of times. Based on the predictive interpretation of the SA Twitter graph, the presence or absence of a Twitter user from one of the five hubs is predictive of the presence or absence of this user from almost all the other district municipalities. Furthermore, the presence or absence of an user from almost all the municipalities that are not hubs is predictive of their presence or absence from each of the hubs.

\begin{figure} [!ht] 
\begin{center}
\includegraphics[width=15cm]{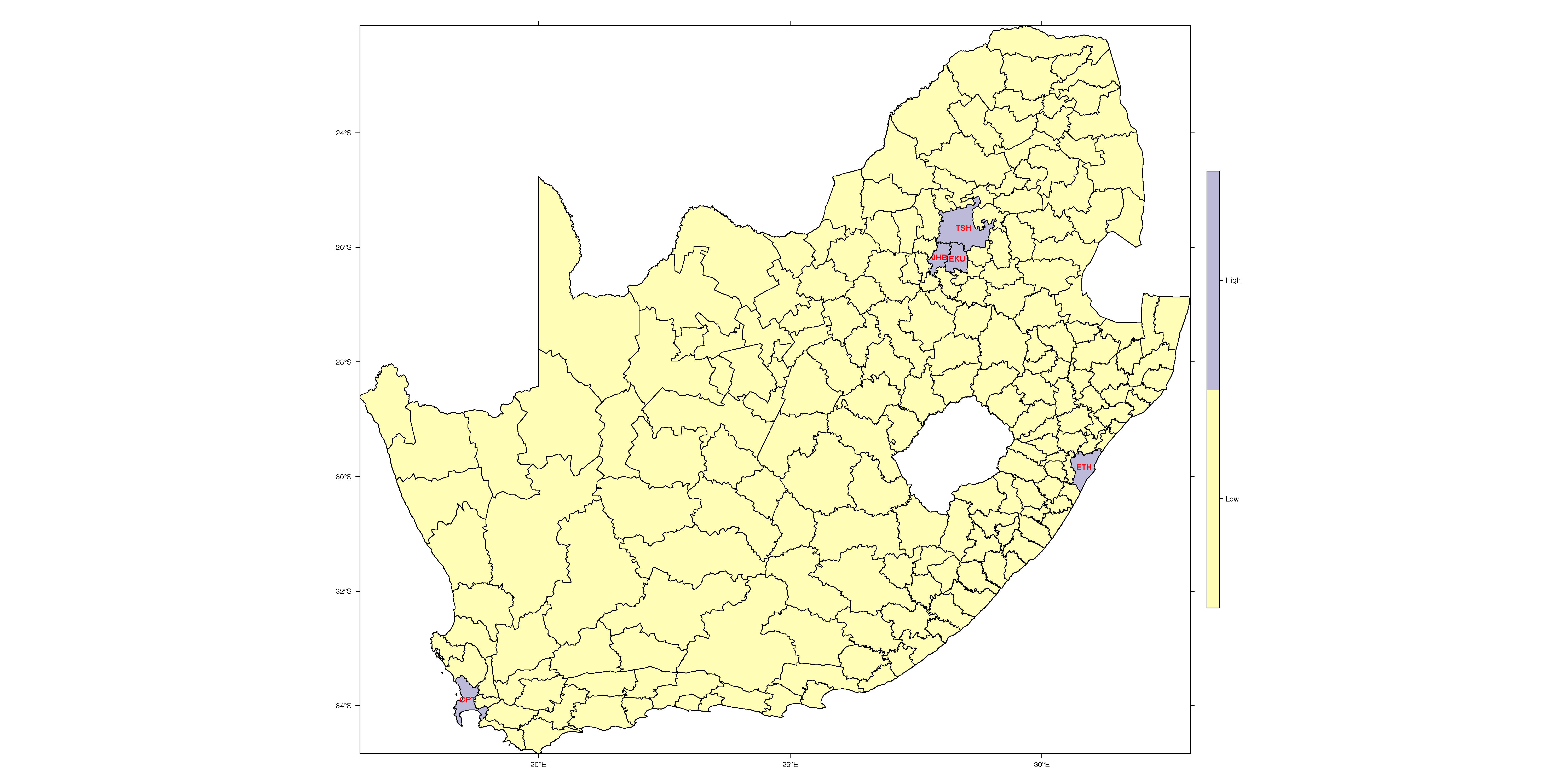}
\end{center}
\caption{ \label{fig:hubmap} 
Map of South Africa showing five municipalities (purple) with the largest centrality measures in the SA Twitter graph. These are the five hubs of this graph. The same municipalities score highest with respect to degree, closeness, betweenness and page rank. The remaining 208 municipalities are shown in yellow.}
\end{figure}

\begin{table}[ht]
\centering
\caption{\label{tab:hub}Summary geographic and demographic information about the five hub municipalities in the SA Twitter graph. Population data extracted from the 2016 Community Survey, Statistics South Africa. Retrieved from \texttt{https://interactive2.statssa.gov.za/webapi}.}
\begin{tabular}{rlllll}
  \hline
  Id. & Municipality name & Province & Area (km$^2$) & Population & Density \\ 
  \hline
TSH & City of Tshwane & Gauteng & 6,298 & 3,275,152 & 520 \\ 
  ETH & eThekwini & KwaZulu-Natal & 2,556 & 3,702,231 & 1,448.50 \\ 
  CPT & City of Cape Town & Western Cape & 2,446 & 4,005,016 & 1,637.60 \\ 
  JHB & City of Johannesburg & Gauteng & 1,645 & 4,949,347 & 3,008.80 \\ 
  EKU & Ekurhuleni & Gauteng & 1,975 & 3,379,104 & 1,710.60 \\
   \hline
\end{tabular}
\end{table}

\begin{figure} [ht] 
\begin{center}
\includegraphics[width=12cm]{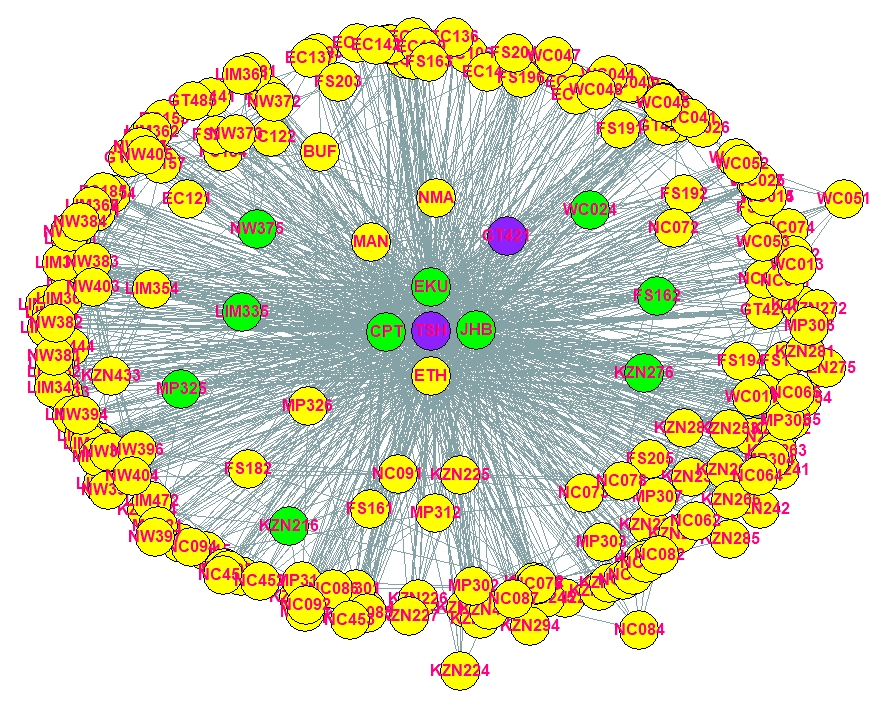}
\end{center}
\caption{ \label{fig:highlownetwork} 
Rendering of the SA Twitter graph that shows the $12$ district municipalities that are linked by an edge with the vertex associated with the Local variable. The two municipalities with an aOR greater than 1 are shown in purple, the $10$ municipalities that have an aOR smaller than 1 are shown in green, while the rest of the municipalities are shown in yellow.}
\end{figure}

The vertex associated with the variable Local is not central in the structure of the SA Twitter graph. Its degree is $12$, and the other centrality measures are also significantly smaller compared to those of the five hubs. The 12 district municipalities that are connected with an edge with the Local vertex are mapped in Figure \ref{fig:foreignmap}; the corresponding vertices in the SA Twitter graph are shown in Figure \ref{fig:highlownetwork}. In addition to four of the hubs (JHB, EKU, TSH, CPT), the presence or absence patterns of a Twitter user from the following municipalities are predictive of whether this user is local to South Africa: FS162 (Kopanong, Free State), KZN216 and KZN276 (Ray Nkonyeni and Big Five Hlabisa, KwaZulu-Natal), LIM335 (Maruleng, Limpopo), MP325 (Bushbuckridge, Mpumalanga), NW375 (Moses Kotane, North West), GT421 (Emfuleni, Gauteng), and WC024 (Stellenbosch, Western Cape). It is quite interesting to examine the spatial distribution of these 12 municipalities: CPT and WC024 are adjacent municipalities around Cape Town; TSH, JHB and EKU define a spatially contiguous region in the Johannerburg/Pretoria area; while LIM335 and MP325 are adjacent municipalities at the border between South Africa and Mozambique. In the KwaZulu-Natal province, the municipalities KZN216 and KZN276 that are located to the south and to the north of the city of Durban are among the neighbors of Local, but the ETH municipality in which Durban is located is not (quite surprisingly) among the neighbors of Local. The FS162 municipality is located south of Bloemfontein \--- a major city in South Africa known for its mining industry. The NW375 municipality is located north west of the Johannesburg/Pretoria area, and it comprises Sun City and a major national park \--- both key touristic destinations. 

\begin{table}[ht]
\centering
\caption{\label{tab:aORlow}Summary geographic and demographic information about the 10 municipalities linked by an edge with the Local (yes/no) variable in the SA Twitter graph that have an aOR smaller than 1. Population data extracted from the 2016 Community Survey, Statistics South Africa. Retrieved from \texttt{https://interactive2.statssa.gov.za/webapi}.}
{\tiny
\begin{tabular}{rllllll}
  \hline
  Id. & Municipality name & Province & Area (km$^2$) & Population & Density & aOR (95\% CI)\\ 
  \hline
  FS162 & Kopanong & Free State & 15,645 & 49,999 & 3.2 & 0.259  (0.237,0.284)\\ 
  KZN216 & Ray Nkonyeni & KwaZulu-Natal & 1,487 & 348,533 & 234.4 & 0.731  (0.6769839,0.790)\\ 
  LIM335 & Maruleng & Limpopo & 3,563 & 99,946 & 28.1 & 0.251 (0.225,0.280)\\ 
  NW375 & Moses Kotane & North West & 5,726 & 243,648 & 42.5 & 0.692  (0.640,0.748)\\ 
  KZN276 & Big Five Hlabisa & KwaZulu-Natal & 3,466 & 116,622 & 33.6 & 0.482  (0.419,0.555)\\ 
  WC024 & Stellenbosch & Western Cape & 831 & 173,197 & 208.4 & 0.760 (0.726,0.796)\\ 
  MP325 & Bushbuckridge & Mpumalanga & 10,248 & 546,215 & 53.3 & 0.413  (0.380,0.448)\\ 
  CPT & City of Cape Town & Western Cape & 2,446 & 4,005,016 & 1,637.60 & 0.354 (0.346,0.362)\\ 
  JHB & City of Johannesburg & Gauteng & 1,645 & 4,949,347 & 3,008.80 & 0.920  (0.898,0.942)\\ 
  EKU & Ekurhuleni & Gauteng & 1,975 & 3,379,104 & 1,710.60 & 0.804  (0.781,0.827)\\ 
   \hline
\end{tabular}
}
\end{table}

\begin{table}[ht]
\centering
\caption{\label{tab:aORhigh}Summary geographic and demographic information about the 2 municipalities linked by an edge with the Local (yes/no) variable in the SA Twitter graph that have an aOR greater than $1$. Population data extracted from the $2016$ Community Survey, Statistics South Africa. Retrieved from \texttt{https://interactive2.statssa.gov.za/webapi}.}
{\tiny
\begin{tabular}{rllllll}
  \hline
 Id. & Municipality name & Province & Area (km$^2$) & Population & Density & aOR (95\% CI)\\ 
  \hline
  TSH & City of Tshwane & Gauteng & 6,298 & 3,275,152 & 520 & 2.347  (2.260,2.437)\\ 
  GT421 & Emfuleni & Gauteng & 966 & 733,445 & 759.3 & 5.414  (4.714,6.218)\\    \hline
\end{tabular}
}
\end{table}

We determine the effect of the presence and absence patterns of Twitter users from these $12$ municipalities on the odds of being local to South Africa by fitting a logistic regression model for the Local variable with $12$ explanatory variables associated with these municipalities. The estimated adjusted odds ratios are given in Tables \ref{tab:aORlow} and \ref{tab:aORhigh}. A number of 10 municipalities have adjusted odds ratios significantly smaller than $1$ at  significance level $\alpha = 0.05$. Given the same presence and absence pattern in the remaining 11 municipalities, a Twitter user that posted geotweets from one of these municipalities has smaller odds of being local to South Africa compared to another Twitter user that did not post geotweets from that municipality. However, the TSH and GT421 municipalities located to the north and to the south of the Johannesburg/Pretoria area have estimated adjusted odds ratios significantly greater than $1$ at significance level $\alpha = 0.05$. Given the same presence and absence pattern in the remaining 11 municipalities, the odds of being local to South Africa of a Twitter user that was present in GT421 (TSH) are $5.414$ ($2.347$) times larger than the odds of being local to South Africa of another Twitter user that was absent from GT421 (TSH). It is known that a considerable number of Mozambicans come to work in the mines in the Johannesburg/Pretoria area for extended periods of time \cite{Baltazar2015}. Their residences might be located in the GT421 and TSH municipalities where they could exceed the number of South African Twitter users.

\begin{figure} [!ht] 
\begin{center}
\includegraphics[width=9cm]{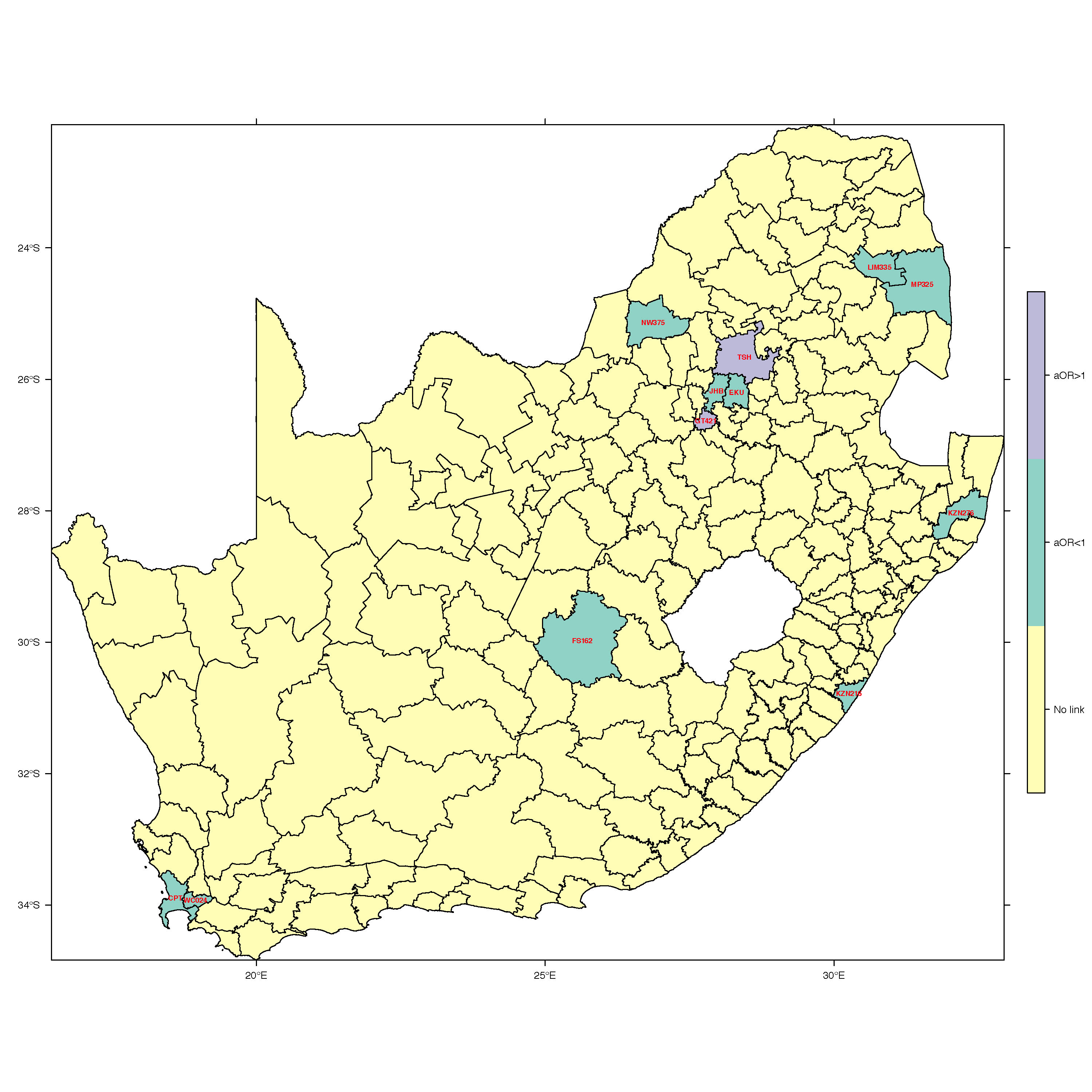}
\end{center}
\caption{ \label{fig:foreignmap} 
Map of South Africa showing the 12 municipalities that are linked by an edge with the vertex associated with the Local variable in the SA Twitter graph.  Two of these municipalities (purple) have an adjusted odds ratio (aOR) greater than 1, while 10 of them (green) have an aOR smaller than 1. The remaining 208 municipalities are shown in yellow.}
\end{figure}

\section{Conclusions} \label{sec:conclusions}

This paper makes several contributions. First, it generalizes the birth-death Markov chain Monte Carlo (BDMCMC) algorithm introduced by \cite{mohammadi2015bayesianStructure} in the context of Gaussian graphical models to general undirected graphical models. Second, based on marginal pseudo-likelihood for categorical data of \cite{pensar2016marginal}, we show how to efficiently calculate the birth and death rates for the BDMCMC algorithm for arbitrary undirected graphs. Third, we use our methodology to analyze a $214$-dimensional contingency table that captures the mobility patterns of Twitter users in South Africa. This is a dataset we collected at the University of Washington which has never been analyzed before. 

The hill-climbing (HC) algorithm \cite{pensar2016marginal} determines graphs with high posterior probability using a greedy hill-climbing optimization algorithm. For this reason, the HC algorithm will inevitably end up in a local maximum. Which local maximum the HC algorithm will find depends on the choice of starting graph. The results of the simulation study from Section \ref{sec:simulation} were obtained by starting the HC algorithm from empty graphs. Since the true graphs were sparse, the HC algorithm recorded a good performance that was comparable with the performance of the BDMCMC algorithm. However, if we would have started the HC algorithm from random graphs that contained a larger number of edges, the HC algorithm might have been at a disadvantage. As we illustrated in  Section \ref{sec:twitteranalysis}, starting the BDMCMC algorithm from sparser or denser graphs led to the identification of the same neighborhood of graphs with high posterior probabilities. The BDMCMC algorithm has a key advantage over the HC algorithm in terms of its ability to visit graphs with lower posterior probability in order to escape local optima, and move towards other graphs with larger posterior probabilities.

Our applied results give an understanding of the movements of $476,601$ individuals that used geolocated tweets in South Africa between $2011$ and $2016$. It is true that the movements of this specific group of people cannot be considered to be representative of major flows of movement of South Africans, or of the visitors of this country. And, due to the selected locations Twitter users choose to post their tweets from, it is possible that even the travel trajectories of these individuals could be only partially captured. However, to the best of our knowledge, there is no other study on human mobility that involves a larger number of individuals in South Africa, and comprises a larger number of recorded locations ($>46$ millions). While our findings must be interpreted with care from a sociodemographic perspective, the methodology we introduce in this article can be successfully applied to modeling patterns of repeated across regions movement that span entire countries, and comprise a large number of individuals. Our modeling approach is based on multi-way tables that cross-classify presence/absence patterns from regions of interest, together with other relevant categorical factors. Our framework goes beyond methods that focus exclusively on modeling flows of migration between origin and destination areas.

We showed that our version of the BDMCMC algorithm can efficiently determine conditional independence graphs for $214$ categorical variables. As far as we know, this is the largest categorical dataset analyzed so far with loglinear models. These developments would not have been possible without Steve Fienberg's visionary life long work which led to the birth of a research community that spans several disciplines (social sciences, health and medical sciences, computer science, and statistics) and will without doubt generate fundamental scientific knowledge for many generations to come.

\section*{Acknowledgment}

This work was supported in part by the National Science Foundation Grant DMS/MPS-1737746 to University of Washington. The authors thank Johan Pensar for providing some of the code used in the simulation study and Sven Baars for his suggestions related to parallel coding in \texttt{C++}. We also would like to thank the Center for Information Technology of the University of Groningen for their support  and for providing access to the Peregrine high performance computing cluster. 

{\tiny
\begin{longtable}{rlllll}
\caption{\label{tab:summarygeodem}Summary geographic and demographic information $213$ South African municipalities. Population data extracted from the 2016 Community Survey, Statistics South Africa. Retrieved from \texttt{https://interactive2.statssa.gov.za/webapi}.}\\
  \hline Id. & Municipality name & Province & Area (km$^2$) & Population & Density\\ \hline
 \endfirsthead   
  
\multicolumn{6}{c}%
{{\bfseries \tablename\ \thetable{} -- continued from previous page}} \\  
\hline Id. & Municipality name & Province & Area (km$^2$) & Population & Density\\ \hline
\endhead
    
EC101 & Dr Beyers Naude & Eastern Cape & 28,653 & 82,197 & 2.9 \\ 
  EC104 & Makana & Eastern Cape & 4,376 & 82,060 & 18.8 \\ 
  EC105 & Ndlambe & Eastern Cape & 1,841 & 63,180 & 34.3 \\ 
  EC121 & Mbhashe & Eastern Cape & 3,303 & 277,250 & 84 \\ 
  EC122 & Mnquma & Eastern Cape & 3,137 & 246,813 & 78.7 \\ 
  EC131 & Inxuba Yethemba & Eastern Cape & 11,663 & 70,493 & 6 \\ 
  EC137 & Engcobo & Eastern Cape & 2,484 & 162,014 & 65.2 \\ 
  EC141 & Elundini & Eastern Cape & 5,019 & 144,929 & 28.9 \\ 
  EC142 & Senqu & Eastern Cape & 7,329 & 140,720 & 19.2 \\ 
  EC153 & Ngquza Hill & Eastern Cape & 2,477 & 303,379 & 122.5 \\ 
  EC154 & Port St Johns & Eastern Cape & 1,291 & 166,779 & 129.2 \\ 
  EC155 & Nyandeni & Eastern Cape & 2,474 & 309,702 & 125.2 \\ 
  EC156 & Mhlontlo & Eastern Cape & 2,880 & 189,176 & 65.7 \\ 
  EC157 & King Sabata Dalindyebo & Eastern Cape & 3,019 & 488,349 & 161.8 \\ 
  EC441 & Matatiele & Eastern Cape & 4,352 & 219,447 & 50.4 \\ 
  FS161 & Letsemeng & Free State & 9,828 & 40,044 & 4.1 \\ 
  FS162 & Kopanong & Free State & 15,645 & 49,999 & 3.2 \\ 
  FS163 & Mohokare & Free State & 8,776 & 35,840 & 4.1 \\ 
  FS181 & Masilonyana & Free State & 6,618 & 62,770 & 9.5 \\ 
  FS182 & Tokologo & Free State & 9,326 & 29,149 & 3.1 \\ 
  FS183 & Tswelopele & Free State & 6,524 & 47,373 & 7.3 \\ 
  FS184 & Matjhabeng & Free State & 5,690 & 429,113 & 75.4 \\ 
  FS185 & Nala & Free State & 4,129 & 78,515 & 19 \\ 
  FS191 & Setsoto & Free State & 5,431 & 117,362 & 21.6 \\ 
  FS192 & Dihlabeng & Free State & 4,868 & 140,044 & 28.8 \\ 
  FS193 & Nketoana & Free State & 5,611 & 64,893 & 11.6 \\ 
  FS194 & Maluti a Phofung & Free State & 4,338 & 353,452 & 81.5 \\ 
  FS195 & Phumelela & Free State & 8,196 & 50,054 & 6.1 \\ 
  FS196 & Mantsopa & Free State & 4,291 & 53,525 & 12.5 \\ 
  FS201 & Moqhaka & Free State & 7,925 & 154,732 & 19.5 \\ 
  FS203 & Ngwathe & Free State & 7,055 & 118,907 & 16.9 \\ 
  FS204 & Metsimaholo & Free State & 1,717 & 163,564 & 95.3 \\ 
  FS205 & Mafube & Free State & 3,971 & 57,574 & 14.5 \\ 
  GT481 & Mogale City & Gauteng & 1,342 & 383,864 & 286 \\ 
  GT485 & Rand West City & Gauteng & 1,115 & 265,887 & 238.5 \\ 
  GT484 & Merafong City & Gauteng & 1,631 & 188,843 & 115.8 \\ 
  TSH & City of Tshwane & Gauteng & 6,298 & 3,275,152 & 520 \\ 
  KZN212 & Umdoni & KwaZulu-Natal & 994 & 144,551 & 145.5 \\ 
  KZN214 & uMuziwabantu & KwaZulu-Natal & 1,089 & 108,576 & 99.7 \\ 
  KZN216 & Ray Nkonyeni & KwaZulu-Natal & 1,487 & 348,533 & 234.4 \\ 
  KZN238 & Alfred Duma & KwaZulu-Natal & 3,764 & 356,274 & 94.6 \\ 
  KZN237 & Inkosi Langalibalele & KwaZulu-Natal & 3,399 & 215,182 & 63.3 \\ 
  KZN235 & Okhahlamba & KwaZulu-Natal & 3,971 & 135,132 & 34 \\ 
  KZN241 & Endumeni & KwaZulu-Natal & 1,610 & 76,639 & 47.6 \\ 
  KZN242 & Nqutu & KwaZulu-Natal & 1,962 & 171,325 & 87.3 \\ 
  KZN252 & Newcastle & KwaZulu-Natal & 1,856 & 389,117 & 209.7 \\ 
  KZN253 & Emadlangeni & KwaZulu-Natal & 3,539 & 36,869 & 10.4 \\ 
  KZN254 & Dannhauser & KwaZulu-Natal & 1,707 & 105,341 & 61.7 \\ 
  KZN261 & eDumbe & KwaZulu-Natal & 1,943 & 89,614 & 46.1 \\ 
  KZN262 & uPhongolo & KwaZulu-Natal & 3,110 & 141,247 & 45.4 \\ 
  KZN263 & Abaqulusi & KwaZulu-Natal & 4,314 & 243,795 & 56.5 \\ 
  KZN265 & Nongoma & KwaZulu-Natal & 2,182 & 211,892 & 97.1 \\ 
  KZN266 & Ulundi & KwaZulu-Natal & 3,250 & 205,762 & 63.3 \\ 
  KZN271 & Umhlabuyalingana & KwaZulu-Natal & 4,977 & 172,077 & 34.6 \\ 
  KZN272 & Jozini & KwaZulu-Natal & 3,442 & 198,215 & 57.6 \\ 
  KZN275 & Mtubatuba & KwaZulu-Natal & 1,970 & 202,176 & 102.6 \\ 
  KZN284 & uMlalazi & KwaZulu-Natal & 2,214 & 223,140 & 100.8 \\ 
  KZN286 & Nkandla & KwaZulu-Natal & 1,828 & 114,284 & 62.5 \\ 
  KZN291 & Mandeni & KwaZulu-Natal & 545 & 147,808 & 271 \\ 
  KZN292 & KwaDukuza & KwaZulu-Natal & 735 & 276,719 & 376.5 \\ 
  KZN293 & Ndwedwe & KwaZulu-Natal & 1,093 & 143,117 & 131 \\ 
  KZN294 & Maphumulo & KwaZulu-Natal & 896 & 89,969 & 100.4 \\ 
  KZN436 & Dr Nkosazana Dlamini Zuma & KwaZulu-Natal & 3,602 & 118,480 & 32.9 \\ 
  KZN434 & Ubuhlebezwe & KwaZulu-Natal & 1,669 & 118,346 & 70.9 \\ 
  LIM331 & Greater Giyani & Limpopo & 4,172 & 256,127 & 61.4 \\ 
  LIM332 & Greater Letaba & Limpopo & 1,891 & 218,030 & 115.3 \\ 
  LIM333 & Greater Tzaneen & Limpopo & 2,897 & 416,146 & 143.7 \\ 
  LIM334 & Ba-Phalaborwa & Limpopo & 7,489 & 168,937 & 22.6 \\ 
  LIM335 & Maruleng & Limpopo & 3,563 & 99,946 & 28.1 \\ 
  LIM355 & Lepele-Nkumpi & Limpopo & 3,484 & 235,380 & 67.6 \\ 
  LIM361 & Thabazimbi & Limpopo & 11,190 & 96,232 & 8.6 \\ 
  LIM362 & Lephalale & Limpopo & 13,794 & 140,240 & 10.2 \\ 
  LIM366 & Bela-Bela & Limpopo & 3,406 & 76,296 & 22.4 \\ 
  LIM367 & Mogalakwena & Limpopo & 6,156 & 325,291 & 52.8 \\ 
  LIM471 & Ephraim Mogale & Limpopo & 2,011 & 127,168 & 63.2 \\ 
  LIM472 & Elias Motsoaledi & Limpopo & 3,713 & 268,256 & 72.2 \\ 
  LIM473 & Makhuduthamaga & Limpopo & 2,110 & 284,435 & 134.8 \\ 
  MP301 & Chief Albert Luthuli & Mpumalanga & 5,559 & 187,629 & 33.7 \\ 
  MP302 & Msukaligwa & Mpumalanga & 6,016 & 164,608 & 27.4 \\ 
  MP303 & Mkhondo & Mpumalanga & 4,882 & 189,036 & 38.7 \\ 
  MP304 & Dr Pixley Ka Isaka Seme & Mpumalanga & 5,227 & 85,395 & 16.3 \\ 
  MP305 & Lekwa & Mpumalanga & 4,557 & 123,419 & 27.1 \\ 
  MP306 & Dipaleseng & Mpumalanga & 2,645 & 45,232 & 17.1 \\ 
  MP307 & Govan Mbeki & Mpumalanga & 2,955 & 340,091 & 115.1 \\ 
  MP311 & Victor Khanye & Mpumalanga & 1,568 & 84,151 & 53.7 \\ 
  MP312 & Emalahleni & Mpumalanga & 2,678 & 455,228 & 170 \\ 
  MP313 & Steve Tshwete & Mpumalanga & 3,976 & 278,749 & 70.1 \\ 
  MP314 & Emakhazeni & Mpumalanga & 4,736 & 48,149 & 10.2 \\ 
  MP315 & Thembisile & Mpumalanga & 2,384 & 333,331 & 139.8 \\ 
  MP316 & Dr JS Moroka & Mpumalanga & 1,416 & 246,016 & 173.7 \\ 
  MP321 & Thaba Chweu & Mpumalanga & 5,719 & 101,895 & 17.8 \\ 
  MP324 & Nkomazi & Mpumalanga & 4,787 & 410,907 & 85.8 \\ 
  NW371 & Moretele & North West & 1,498 & 191,306 & 127.7 \\ 
  NW372 & Local Municipality of Madibeng & North West & 3,720 & 536,110 & 144.1 \\ 
  NW373 & Rustenburg & North West & 3,416 & 626,522 & 183.4 \\ 
  NW374 & Kgetlengrivier & North West & 3,973 & 59,562 & 15 \\ 
  NW375 & Moses Kotane & North West & 5,726 & 243,648 & 42.5 \\ 
  NW381 & Ratlou & North West & 4,884 & 106,108 & 21.7 \\ 
  NW382 & Tswaing & North West & 5,875 & 129,052 & 22 \\ 
  NW383 & Mafikeng & North West & 3,646 & 314,394 & 86.2 \\ 
  NW384 & Ditsobotla & North West & 6,387 & 181,865 & 28.5 \\ 
  NW385 & Ramotshere Moiloa & North West & 7,323 & 157,690 & 21.5 \\ 
  NW392 & Naledi & North West & 7,030 & 68,803 & 9.8 \\ 
  NW393 & Mamusa & North West & 3,614 & 64,000 & 17.7 \\ 
  NW394 & Greater Taung & North West & 5,639 & 167,827 & 29.8 \\ 
  NW396 & Lekwa-Teemane & North West & 3,654 & 56,025 & 15.3 \\ 
  NW397 & Kagisano/Molopo & North West & 23,827 & 102,703 & 4.3 \\ 
  NW403 & City of Matlosana & North West & 3,602 & 417,282 & 115.8 \\ 
  NW404 & Maquassi Hills & North West & 4,671 & 82,012 & 17.6 \\ 
  NC061 & Richtersveld & Northern Cape & 9,608 & 12,487 & 1.3 \\ 
  NC062 & Nama Khoi & Northern Cape & 17,990 & 46,512 & 2.6 \\ 
  NC064 & Kamiesberg & Northern Cape & 14,208 & 9,605 & 0.7 \\ 
  NC065 & Hantam & Northern Cape & 39,085 & 21,540 & 0.6 \\ 
  NC066 & Karoo Hoogland & Northern Cape & 30,230 & 13,009 & 0.4 \\ 
  NC067 & Kh‰i-Ma & Northern Cape & 15,715 & 12,333 & 0.8 \\ 
  NC071 & Ubuntu & Northern Cape & 20,393 & 19,471 & 1 \\ 
  NC072 & Umsobomvu & Northern Cape & 6,813 & 30,883 & 4.5 \\ 
  NC073 & Emthanjeni & Northern Cape & 13,472 & 45,404 & 3.4 \\ 
  NC074 & Kareeberg & Northern Cape & 17,701 & 12,772 & 0.7 \\ 
  NC075 & Renosterberg & Northern Cape & 5,529 & 11,818 & 2.1 \\ 
  NC076 & Thembelihle & Northern Cape & 8,023 & 16,230 & 2 \\ 
  NC077 & Siyathemba & Northern Cape & 14,727 & 23,075 & 1.6 \\ 
  NC078 & Siyancuma & Northern Cape & 16,753 & 35,941 & 2.1 \\ 
  NC082 & Kai !Garib & Northern Cape & 26,377 & 68,929 & 2.6 \\ 
  NC084 & !Kheis & Northern Cape & 11,107 & 16,566 & 1.5 \\ 
  NC085 & Tsantsabane & Northern Cape & 18,290 & 39,345 & 2.2 \\ 
  NC086 & Kgatelopele & Northern Cape & 2,478 & 20,691 & 8.3 \\ 
  NC091 & Sol Plaatjie & Northern Cape & 3,145 & 255,041 & 81.1 \\ 
  NC092 & Dikgatlong & Northern Cape & 7,316 & 48,473 & 6.6 \\ 
  NC453 & Gamagara & Northern Cape & 2,648 & 53,656 & 20.3 \\ 
  WC011 & Matzikama & Western Cape & 12,981 & 71,045 & 5.5 \\ 
  WC012 & Cederberg & Western Cape & 8,007 & 52,949 & 6.6 \\ 
  WC013 & Bergrivier & Western Cape & 4,407 & 67,474 & 15.3 \\ 
  WC022 & Witzenberg & Western Cape & 10,753 & 130,548 & 12.1 \\ 
  WC023 & Drakenstein & Western Cape & 1,538 & 280,195 & 182.2 \\ 
  WC025 & Breede Valley & Western Cape & 3,834 & 176,578 & 46.1 \\ 
  WC026 & Langeberg & Western Cape & 4,518 & 105,483 & 23.3 \\ 
  WC033 & Cape Agulhas & Western Cape & 3,471 & 36,000 & 10.4 \\ 
  WC034 & Swellendam & Western Cape & 3,835 & 40,211 & 10.5 \\ 
  WC041 & Kannaland & Western Cape & 4,765 & 24,168 & 5.1 \\ 
  WC042 & Hessequa & Western Cape & 5,733 & 54,237 & 9.5 \\ 
  WC043 & Mossel Bay & Western Cape & 2,001 & 94,135 & 47 \\ 
  WC044 & George & Western Cape & 5,191 & 208,237 & 40.1 \\ 
  WC045 & Oudtshoorn & Western Cape & 3,540 & 97,509 & 27.5 \\ 
  WC047 & Bitou & Western Cape & 992 & 59,157 & 59.6 \\ 
  WC048 & Knysna & Western Cape & 1,109 & 73,835 & 66.6 \\ 
  WC051 & Laingsburg & Western Cape & 8,784 & 8,895 & 1 \\ 
  WC052 & Prince Albert & Western Cape & 8,153 & 14,272 & 1.8 \\ 
  WC053 & Beaufort West & Western Cape & 21,917 & 51,080 & 2.3 \\ 
  NC451 & Joe Morolong & Northern Cape & 20,180 & 84,201 & 4.2 \\ 
  NC452 & Ga-Segonyana & Northern Cape & 4,495 & 104,408 & 23.2 \\ 
  KZN213 & Umzumbe & KwaZulu-Natal & 1,221 & 151,676 & 124.2 \\ 
  KZN276 & Big Five Hlabisa & KwaZulu-Natal & 3,466 & 116,622 & 33.6 \\ 
  KZN227 & Richmond & KwaZulu-Natal & 1,231 & 71,322 & 57.9 \\ 
  KZN433 & Greater Kokstad & KwaZulu-Natal & 2,680 & 76,753 & 28.6 \\ 
  KZN435 & Umzimkhulu & KwaZulu-Natal & 2,436 & 197,286 & 81 \\ 
  NC093 & Magareng & Northern Cape & 1,546 & 24,059 & 15.6 \\ 
  NC094 & Phokwane & Northern Cape & 828 & 60,168 & 72.7 \\ 
  WC024 & Stellenbosch & Western Cape & 831 & 173,197 & 208.4 \\ 
  WC031 & Theewaterskloof & Western Cape & 3,259 & 117,167 & 36 \\ 
  EC442 & Umzimvubu & Eastern Cape & 2,579 & 199,620 & 77.4 \\ 
  EC444 & Ntabankulu & Eastern Cape & 1,385 & 128,848 & 93.1 \\ 
  EC443 & Mbizana & Eastern Cape & 2,415 & 319,948 & 132.5 \\ 
  EC123 & Great Kei & Eastern Cape & 1,700 & 31,692 & 18.6 \\ 
  EC124 & Amahlathi & Eastern Cape & 4,505 & 101,826 & 22.6 \\ 
  KZN221 & uMshwathi & KwaZulu-Natal & 1,866 & 111,645 & 59.8 \\ 
  KZN244 & Msinga & KwaZulu-Natal & 2,375 & 184,494 & 77.7 \\ 
  ETH & eThekwini & KwaZulu-Natal & 2,556 & 3,702,231 & 1,448.50 \\ 
  KZN226 & Mkhambathini & KwaZulu-Natal & 868 & 57,075 & 65.7 \\ 
  KZN225 & The Msunduzi & KwaZulu-Natal & 751 & 679,039 & 904.1 \\ 
  KZN222 & uMngeni & KwaZulu-Natal & 1,520 & 109,867 & 72.3 \\ 
  KZN224 & Impendle & KwaZulu-Natal & 1,610 & 29,526 & 18.3 \\ 
  KZN281 & Mfolozi & KwaZulu-Natal & 1,300 & 144,363 & 111.1 \\ 
  KZN282 & uMhlathuze & KwaZulu-Natal & 1,233 & 410,465 & 332.8 \\ 
  KZN285 & Mthonjaneni & KwaZulu-Natal & 1,639 & 78,883 & 48.1 \\ 
  EC106 & Sundays River Valley & Eastern Cape & 5,995 & 59,793 & 10 \\ 
  EC108 & Kouga & Eastern Cape & 2,670 & 112,941 & 42.3 \\ 
  EC109 & Kou-Kamma & Eastern Cape & 3,642 & 43,688 & 12 \\ 
  NMA & Nelson Mandela Bay & Eastern Cape & 1,957 & 1,263,051 & 645.4 \\ 
  BUF & Buffalo City & Eastern Cape & 2,750 & 834,997 & 303.6 \\ 
  EC126 & Ngqushwa & Eastern Cape & 2,115 & 63,694 & 30.1 \\ 
  MP325 & Bushbuckridge & Mpumalanga & 10,248 & 546,215 & 53.3 \\ 
  EC135 & Intsika Yethu & Eastern Cape & 2,873 & 152,159 & 53 \\ 
  EC136 & Emalahleni & Eastern Cape & 3,484 & 124,532 & 35.7 \\ 
  EC138 & Sakhisizwe & Eastern Cape & 2,318 & 63,846 & 27.5 \\ 
  WC014 & Saldanha Bay & Western Cape & 2,015 & 111,173 & 55.2 \\ 
  WC015 & Swartland & Western Cape & 3,707 & 133,762 & 36.1 \\ 
  WC032 & Overstrand & Western Cape & 1,675 & 93,407 & 55.8 \\ 
  CPT & City of Cape Town & Western Cape & 2,446 & 4,005,016 & 1,637.60 \\ 
  LIM351 & Blouberg & Limpopo & 9,540 & 172,601 & 18.1 \\ 
  LIM353 & Molemole & Limpopo & 3,628 & 125,327 & 34.5 \\ 
  LIM354 & Polokwane & Limpopo & 5,054 & 797,127 & 157.7 \\ 
  LIM368 & Modimolle/Mookgophong & Limpopo & 10,367 & 107,699 & 10.4 \\ 
  LIM476 & Greater Tubatse/Fetakgomo & Limpopo & 5,693 & 489,902 & 86 \\ 
  LIM341 & Musina & Limpopo & 10,347 & 132,009 & 12.8 \\ 
  LIM343 & Thulamela & Limpopo & 2,642 & 497,237 & 188.2 \\ 
  LIM344 & Makhado & Limpopo & 7,605 & 416,728 & 54.8 \\ 
  LIM345 & New & Limpopo & 5,003 & 347,974 & 69.6 \\ 
  NC087 & Dawid Kruiper & Northern Cape & 44,231 & 107,161 & 2.4 \\ 
  MP326 & Mbombela & Mpumalanga & 7,141 & 695,913 & 97.4 \\ 
  NW405 & Ventersdorp/Tlokwe & North West & 6,398 & 243,527 & 38.1 \\ 
  MAN & Mangaung & Free State & 9,886 & 787,803 & 79.7 \\ 
  EC145 & Walter Sisulu & Eastern Cape & 13,269 & 87,263 & 6.6 \\ 
  EC139 & Enoch Mgijima & Eastern Cape & 13,584 & 267,011 & 19.7 \\ 
  EC129 & Raymond Mhlaba & Eastern Cape & 6,357 & 159,515 & 25.1 \\ 
  KZN245 & Umvoti & KwaZulu-Natal & 2,705 & 122,423 & 45.3 \\ 
  KZN223 & Mpofana & KwaZulu-Natal & 1,757 & 37,391 & 21.3 \\ 
  GT423 & Lesedi & Gauteng & 1,484 & 112,472 & 75.8 \\ 
  GT422 & Midvaal & Gauteng & 1,722 & 111,612 & 64.8 \\ 
  GT421 & Emfuleni & Gauteng & 966 & 733,445 & 759.3 \\ 
  EC102 & Blue Crane Route & Eastern Cape & 11,068 & 36,063 & 3.3 \\ 
  JHB & City of Johannesburg & Gauteng & 1,645 & 4,949,347 & 3,008.80 \\ 
  EKU & Ekurhuleni & Gauteng & 1,975 & 3,379,104 & 1,710.60 \\ 
   \hline
\end{longtable}
}


\end{document}